\documentclass[a4paper,onecolumn,unpublished,amsfonts,amsmath,amssymb,accepted=2024-08-19]{quantumarticle}
\pdfoutput=1
\usepackage[utf8]{inputenc}
\usepackage[english]{babel}
\usepackage[T1]{fontenc}
\usepackage{soul,braket,comment,enumitem,mathtools,bm}

\usepackage[unicode=true,
bookmarks=true,bookmarksopen=false,
breaklinks=false,pdfborder={0 0 0},colorlinks=true]
{hyperref}
\usepackage{xcolor}
\definecolor{cblue}{rgb}{0.16, 0.32, 0.75}
\definecolor{cred}{rgb}{0.7, 0.11, 0.11}
\hypersetup{%
	,linkcolor=cred
	,citecolor=cblue
	,urlcolor=cblue
}

\usepackage{amsthm}
\newtheorem{theorem}{Theorem}[section]

\newtheorem{lemma}[theorem]{Lemma}
\newtheorem{proposition}[theorem]{Proposition}
\theoremstyle{definition}
\newtheorem{definition}[theorem]{Definition}
\newtheorem{assumption}[theorem]{Assumption}
\theoremstyle{remark}
\newtheorem{remark}[theorem]{Remark}

\numberwithin{equation}{section}

\newcommand{\triang}{\hfill$\triangle$}
\AtEndEnvironment{example}{\triang}
\AtEndEnvironment{remark}{\triang}

\usepackage[numbers]{natbib}
\usepackage{dsfont}
\usepackage[makeroom]{cancel}
\usepackage[allowbreakbefore,nospacearound,shortcuts]{extdash}

\newcommand{\e}{\mathrm{e}}
\renewcommand{\i}{\mathrm{i}}

\usepackage[normalem]{ulem}

\allowdisplaybreaks		
					
\begin{document}

\title{Double or nothing: a Kolmogorov extension theorem for multitime (bi)probabilities in quantum mechanics}

\author{Davide Lonigro}
\orcid{0000-0002-0792-8122}
\address{Department Physik, Friedrich-Alexander-Universität Erlangen-Nürnberg, Staudtstraße 7, 91058 Erlangen, Germany}
\address{Dipartimento di Matematica, Università degli Studi di Bari Aldo Moro, via E. Orabona 4, 70125 Bari, Italy}
\address{Istituto Nazionale di Fisica Nucleare, Sezione di Bari, via G. Amendola 173, 70126 Bari, Italy}
%\email{davide.lonigro@fau.de}

\author{Fattah Sakuldee}
\orcid{0000-0001-8756-7904}
\address{Wilczek Quantum Center, School of Physics and Astronomy, Shanghai Jiao Tong University, 800 Dongchuan Road, Minhang, 200240 Shanghai, China}
\address{The International Centre for Theory of Quantum Technologies, University of Gda\'nsk, Jana Ba\.zy\'nskiego 1A, 80-309 Gda\'nsk, Poland}
%\email{fattah.sakuldee@sjtug.edu.cnpl}

\author{{\L}ukasz Cywi{\'n}ski}
\orcid{0000-0002-0162-7943}
\address{Institute of Physics, Polish Academy of Sciences, al.~Lotnik{\'o}w 32/46, PL 02-668 Warsaw, Poland}
%\email{lcyw@ifpan.edu.pl}

\author{Dariusz Chru\'{s}ci\'{n}ski}
\orcid{0000-0002-6582-6730}
\address{Institute of Physics, Faculty of Physics, Astronomy and Informatics, Nicolaus Copernicus University,
Grudziadzka 5/7, 87-100 Toru\'n, Poland}
%\email{darch@fizyka.umk.pl}

\author{Piotr Sza{\'n}kowski}
\orcid{0000-0003-4306-8702}
\address{Institute of Physics, Polish Academy of Sciences, al.~Lotnik{\'o}w 32/46, PL 02-668 Warsaw, Poland}
\email{piotr.szankowski@ifpan.edu.pl}

\maketitle

\begin{abstract}
    The multitime probability distributions obtained by repeatedly probing a quantum system via the measurement of an observable generally violate Kolmogorov's consistency property. Therefore, one cannot interpret such distributions as the result of the sampling of a single trajectory. We show that, nonetheless, they do result from the sampling of one \emph{pair} of trajectories. In this sense, rather than give up on trajectories, quantum mechanics requires to double down on them. To this purpose, we prove a generalization of the Kolmogorov extension theorem that applies to families of complex-valued bi-probability distributions (that is, defined on pairs of elements of the original sample spaces), and we employ this result in the quantum mechanical scenario. We also discuss the relation of our results with the quantum comb formalism.
\end{abstract}

\section{Introduction}\label{sec:intro}

The theory of stochastic processes finds one of its cornerstones in Kolmogorov's extension theorem~\cite{Kolmogorov}. Under suitable assumptions, the theorem ensures that any \textit{consistent} family of multitime joint probabilities---typically used to describe the distribution of measurement outcomes of a physical system repeatedly probed by the experimenter---arises from an underlying ``master probability'' on a space of \textit{trajectories}. Here, focusing on the case of discrete-valued outcomes in a finite set $\Omega$, ``consistent'' means the following: given any $n$-tuple of times $t_1< t_2<\ldots< t_n$, one has
\begin{equation}\label{eq:kolmogorov_consistency}
    \sum_{f_j\in \Omega}P_{t_n,\ldots,t_j,\dots,t_1}(f_n,\ldots,f_j,\ldots,f_1)=P_{t_n,\ldots,\cancel{t_j},\ldots,t_1}(f_n,\ldots,\cancel{f_j},\ldots,f_1),
\end{equation}
the two functions in the above equation respectively representing the $n$-time and $(n-1)$-time probability distributions. The extension theorem then asserts the existence of a probability measure $\mathcal{P}[f]\,[\mathcal{D}f]$ on the space of trajectories $t\mapsto f(t)\in\Omega$, from which all multitime probability distributions can be obtained as restrictions to discrete time grids. In terms of functional integrals, this corresponds to the formal relation
\begin{equation}\label{eq:functional_integral}
    P_{t_n,\ldots,t_1}(f_n,\ldots,f_1)=\int\Big(\prod_{j=1}^n\delta_{f(t_j),f_j}\Big)\mathcal{P}[f][\mathcal{D}f],
\end{equation}
that is, the $n$-time probability $P_{t_n,\ldots,t_1}(f_n,\ldots,f_1)$ is the outcome of summing over all trajectories $f(t)$ that pass through the specified values at the corresponding times, $f(t_1)=f_1,\ldots,f(t_n)=f_n$. Equivalently, all multitime distributions can be obtained by an underlying \textit{stochastic process}. The phenomenological importance of the extension theorem cannot be understated: it guarantees that, when consistent, the measured chronological sequences are nothing but discrete-time samples of an underlying trajectory that is traced over time in accordance with the dynamical laws of the system.

Such a simple paradigm is a defining feature of classical theories, but it is nowhere to be found in quantum mechanics\footnote{
Generalizations of the Kolmogorov consistency condition for a quantum counterpart of stochastic processes, known as ``quantum combs'' or ``process matrices'', were already investigated in the literature on multitime quantum processes
\cite{Milz2020,Chiribella2009}. In principle, such a construction is based on a different generalization of the trajectory picture than in our work, and hence the resulting analyses are given in different setups. We will revisit the connection of our results to the quantum comb formalism in Section~\ref{subsec:quantum_combs}.
}.
The multitime probability distributions obtained by repeatedly probing a quantum system do \textit{not}, in general, satisfy the consistency property~\eqref{eq:kolmogorov_consistency}. The standard Hilbert space--based formalism of quantum mechanics reflects this basic observation. A quantum system initially prepared in a state $\hat{\rho}$, evolving via a unitary propagator $\hat{U}_{t,t_0}$ generated by a (possibly time-dependent) Hamiltonian $\hat{H}(t)$, and periodically probed with a measuring device represented by a family of orthogonal projectors $\{\hat{P}(f)\}_{f\in\Omega}$ (physically, the spectral resolution of the probed observable) yields a family of multitime probabilities given by
\begin{equation}\label{eq:quantum_prob_0}
    P_{t_n,\ldots,t_1}(f_n,\dots,f_1)=\operatorname{tr}\bigg[\Big(\prod_{j=n}^1\hat{P}_{t_j}(f_j)\Big)\hat{\rho}
    \Big(\prod_{j=1}^n\hat{P}_{t_j}(f_j)\Big)\bigg],
\end{equation}
where $\hat{P}_t(f)=\hat{U}_{0,t}\hat{P}(f)\hat{U}_{t,0}$ and $\prod_{j=n}^1\hat A_j$ ($\prod_{j=1}^n\hat A_j)$ indicates an ordered composition, $\hat A_n\cdots\hat A_1$ ($\hat A_1\cdots\hat A_n$). From Eq.~\eqref{eq:quantum_prob_0} one easily shows that the consistency~\eqref{eq:kolmogorov_consistency} is generally violated, albeit it can be satisfied for specific systems~\cite{Smirne_2019,Strasberg_PRA19,lonigro2022classicality,chruscinski2023markovianity,Strasberg_SciPost23,Szankowski_2024}. 
Therefore, in general there is no underlying stochastic process generating the family of probability distributions representing the phenomenology of the system---in other words, we cannot interpret quantum mechanical probabilities as the result of the probing of one underlying, objective trajectory. 

\subsection{Main result}\label{subsec:main_result}

The violation of consistency by the distributions~\eqref{eq:quantum_prob_0}, and the consequent incompatibility with the Kolmogorov extension theorem, forces the quantum theory to abandon the trajectory picture interpretation, that can be instead regarded as the foundational tenet of classical theories. Indeed, there is no way to reconstruct the outcomes of a quantum mechanical experiment from \textit{one} trajectory. In this paper we will show that, nevertheless, \textit{two} trajectories are enough for this purpose.

The basic idea is remarkably simple. Instead of the probabilities~\eqref{eq:quantum_prob_0}, we shall consider the following family of complex-valued functions defined on the space of \textit{pairs} of $n$-tuples,
\begin{equation}\label{eq:quantum_biprob_0}
    Q_{t_n,\ldots,t_1}(f_n^+,f_n^-;\ldots,f_1^+,f_1^-)=\operatorname{tr}\bigg[\Big(\prod_{j=n}^1\hat{P}_{t_j}(f^+_j)\Big)\hat{\rho}
    \Big(\prod_{j=1}^n\hat{P}_{t_j}(f^-_j)\Big)\bigg].
\end{equation}
They are \textit{normalized}, $\sum_{f^\pm_n,\ldots,f_1^\pm}Q_{t_n,\ldots,t_1}(f_n^+,f_n^-;\ldots;f_1^+,f_1^-)=1$, and their ``diagonal'' values coincide with the probability distributions~\eqref{eq:quantum_prob_0}, $Q_{t_n,\ldots,t_1}(f_n,f_n;\ldots;f_1,f_1) = P_{t_n,\ldots,t_1}(f_n,\ldots,f_1)$. As such, we will refer to these functions as \textit{bi-probability} distributions, even though $Q_{t_n,\ldots,t_1}$, being complex-valued, are not proper probability distributions themselves. A simple calculation shows that the family of bi-probabilities \textit{do} satisfy a consistency--like condition,
\begin{equation}\label{eq:kolmogorov_biconsistency}
    \sum_{f_j^+,f_j^-}\!\!\!Q_{t_n,\ldots,t_j,\ldots,t_1}(f_n^+,f_n^-;...;f_j^+,f_j^-;...;f_1^+,f_1^-)=Q_{t_n,\ldots,\cancel{t_j},\ldots,t_1}(f_n^+,f_n^-;\ldots;\cancel{f_j^+,f_j^-};\ldots;f_1^+,f_1^-).
\end{equation}
We will denote this property as \textit{bi-consistency} to stress the fact that the sum has to be performed \textit{separately} on the two $n$-tuples $(f_n^+,\ldots,f_1^+)$ and $(f_n^-,\ldots,f_1^-)$. 

Motivated by this observation, we will prove that a \textit{generalization} of Kolmogorov's theorem applies to quantum mechanics: there exists a master, complex-valued bi-probability measure $\mathcal{Q}[f^+,f^-][\mathcal{D}f^+][\mathcal{D}f^-]$ on the space of trajectory \textit{pairs} $t\mapsto (f^+(t),f^-(t))$, such that a relation analogous to Eq.~\eqref{eq:functional_integral} holds:
\begin{equation}\label{eq:functional_biintegral}
    Q_{t_n,\dots,t_1}(f_n^+,f_n^-;\ldots;f_1^+,f_1^-)=\iint\Big(\prod_{j=1}^n
        \delta_{f^+(t_j),f^+_j}\,\delta_{f^-(t_j),f^-_j}
    \Big)\mathcal{Q}[f^+,f^-][\mathcal{D}f^+][\mathcal{D}f^-],
\end{equation}
whence, in particular,
\begin{equation}\label{eq:functional_biintegral2}
    P_{t_n,\dots,t_1}(f_n;\dots;f_j)=\iint\Big(\prod_{j=1}^n
            \delta_{f^+(t_j),f_j}\,\delta_{f^-(t_j),f_j}
    \Big)\mathcal{Q}[f^+,f^-][\mathcal{D}f^+][\mathcal{D}f^-],
\end{equation}
and the $n$-time probability $P_{t_n,\ldots,t_1}(f_n,\ldots,f_j)$ is an outcome of the superposition of all \textit{pairs} of trajectories $(f^+(t),f^-(t))$ such that \textit{both} pass through the same sequence of values $f_1,\dots,f_n$ at the corresponding times $t_1,\dots,t_n$. We refer to Theorem~\ref{thm:quantum_extension} for the precise statement. As a ``bonus'', the link between quantum mechanics and the classical theory of stochastic processes becomes manifest: Eq.~\eqref{eq:functional_biintegral2} reduces to its classical counterpart~\eqref{eq:functional_integral} in the limit in which the only pairs of trajectories $(f^+,f^-)$ that contribute to the integral are those with $f^+(t)=f^-(t)$ at all times.
%---\add{that is, there is \textit{no interference} between trajectories.} \com{DL: let me know if you agree with this additional sentence; if not, or if you're unsure, just remove it.} \com{FS: I think it's better without. We barely discuss interference anywhere in the manuscript.}

The significance of this result reaches beyond the context of sequential measurements. The bi-probability formalism is a general purpose parameterization that find its use in virtually any context of quantum theory~\cite{Szankowski_SciRep20,Szankowski_PRA21,Szankowski_SciPostLecNotes23,Szankowski_2024}. One illustrative example of such an application is the description of open system dynamics. Let $\hat\rho_{O} \mapsto \Lambda_t\hat\rho_{O}$ be the dynamical map for an open system $O$ that is coupled to an observable $\sum_{f\in\Omega} f\hat P(f) = \hat F$ via the Hamiltonian $\hat H_{OE}(t) = \hat H_{O}\otimes\hat 1+ \hat 1\otimes \hat H(t) + \lambda \hat V_{O}\otimes\hat F$, where the original system plays the role of the environment. Then $\Lambda_t$ can be parameterized as a ``bi-average'' over trajectory pairs~\cite{Szankowski_SciRep20,Szankowski_SciPostLecNotes23,Szankowski_2024},
\begin{align}
\nonumber
    \Lambda_t\hat\rho_{O} &=\operatorname{tr}_{E}\Big[\big(\mathcal{T}\e^{-\i\int_0^t\hat H_{OE}(s)\mathrm{d}s}\big)\hat\rho_{O}\otimes\hat\rho\,
        \big(\mathcal{T}\e^{-\i\int_0^t\hat H_{OE}(s)\mathrm{d}s}\big)^\dagger\Big]\\
    &= \iint
        \e^{-\i \int_0^t\big(\hat H_{O} + \lambda f^+(s)\hat V_{O}\big)\mathrm{d}s}\hat\rho_{O}\,\e^{+\i \int_0^t\big(\hat H_{O} + \lambda f^-(s)\hat V_{O}\big)\mathrm{d}s}\,\mathcal{Q}[f^+,f^-][\mathcal{D}f^+][\mathcal{D}f^-].
\end{align}
As such, the extension theorem proven here provides the rigorous mathematical justification for the use of the \textit{bi-trajectory picture} in the bi-probability formulation, a justification that was previously lacking.

\subsection{Strategy and outline}

The goal of the paper is achieved in two steps:
\begin{itemize}
     \item the generalization of the extension theorem to \textit{complex-valued} (bi)probabilities is stated and subsequently proven in Section~\ref{sec:gen_extension_thrm};
    \item the generalized extension theorem is shown in Section~\ref{sec:quantum_extension} to be applicable to quantum mechanical bi-probabilities given by Eq.~\eqref{eq:quantum_biprob_0}. 
\end{itemize}
The first step is made necessary by the fact that Kolmogorov's extension theorem, in its standard formulation (see e.g.~\cite{Davies_76,bhattacharya2017basic}), only handles \textit{positive} and \textit{normalized} distributions. Fortunately, it is possible to adapt the proof of its standard version to the case of complex-valued functions. Notably, an additional requirement with respect to the standard one is needed: for a family of general bi-probabilities $\{B_{t_n,\dots,t_1}:n\in\mathbb{N},t_1<\dots< t_n\}$ to be extendable it must satisfy, on top of bi-consistency, a \textit{uniform bound} condition:
\begin{equation}\label{eq:uniform_bound_intro}
    \sup\Big\{
    \sum_{f_n^+,f_n^-}\cdots\sum_{f_1^+,f_1^-}\left|B_{t_n,\dots,t_1}(f_n^+,f_n^-;\dots,f_1^+,f_1^-)\right|
    \,:\,n\in\mathbb{N},\,t_1<\cdots<t_n
    \Big\}<\infty,
\end{equation}
which, while obvious for normalized positive distributions, is generally nontrivial when signed or complex-valued functions are involved, like in this case. A thorough explanation of the mathematics behind this construction (even beyond the discrete scenario discussed here) is reported in Appendix~\ref{app:measure_theory}.

In the second step we go back to quantum mechanics and prove that the functions defined in Eq.~\eqref{eq:quantum_biprob_0} do satisfy the assumptions of the generalized extension theorem. This will essentially boil down to showing that the crucial property~\eqref{eq:uniform_bound_intro} does hold as long as the family of quantum bi-probabilities is confined to a finite time window. In fact, demonstrating uniform boundedness for quantum mechanical bi-probabilities turns out to be the main technical challenge in the way towards the desired result.

The concluding discussion of our results, their connections with related concepts found in the literature, and possible future developments, are gathered in Section~\ref{sec:discussion}. The summary and final remarks are collected in Section~\ref{sec:conclusion}.

\section{Kolmogorov extension theorem for multitime bi-probabilities}\label{sec:gen_extension_thrm}

Before we investigate the extension theorem for the specific family of bi-probabilities encountered in the quantum setting (cf.~Eq.~\eqref{eq:quantum_biprob_0}), we shall discuss how it operates in the general context. Let $\Omega$ be a set of real numbers, and $I$ an arbitrary (continuous or discrete) index set. Physically, $\Omega$ represents the set of possible outcomes of a single experiment, and $I$ could be regarded as the time interval over which the experiment or the dynamics of the system takes place. In the abstract context considered here, the purpose of $I$ is to keep track of the order of the supplied arguments. Now, consider a family of functions
\begin{align}
    \mathbb{B}_I &:= \{ B_{t_n,\ldots,t_1}(f_n^+,f_n^-;\ldots;f_1^+,f_1^-)\ :\ n\in\mathbb{N},\,t_1,\ldots,t_n\in I,\,t_1<\ldots<t_n \},
\end{align}
where each member is an abstract complex-valued function of a pair of $n$-tuples,
\begin{align}
    \Omega^n\times\Omega^n\ni (\bm{f}^+_n,\bm{f}^-_n)&\mapsto B_{\bm{t}_n}(\bm{f}^+_n,\bm{f}^-_n)\in\mathbb{C}. 
\end{align}
Here and in the following, we will be using the following notation: $n$-tuples of elements either of $\Omega$ or of $I$ are written as quasi-vectors $\bm{f}_n=(f_n,\ldots,f_1)$ and $\bm{t}_n=(t_n,\ldots,t_1)$, with the subscript in $\bm{f}_n$ and $\bm{t}_n$ denoting the number of elements. Even when not explicitly indicated, we shall always assume that the times $t_1,\ldots,t_n\in I$ are distinct and \textit{ordered}, $t_1<t_2<\ldots<t_n$. Besides, with a slight abuse of notation, we will use either of the following symbols to denote the values of $B_{\bm{t}_n}$:
\begin{equation}
    B_{\bm{t}_n}(\bm{f}_n^+,\bm{f}_n^-),\qquad B_{\bm{t}_n}(f_n^+,f_n^-;\ldots;f_1^+,f_1^-),\qquad B_{t_n,\ldots,t_1}(f_n^+,f_n^-;\ldots;f_1^+,f_1^-),
\end{equation}
the second or third one being used whenever we will need to single out one particular pair of entries.

Presently the goal is to examine the specific mechanisms behind the extension theorem for the family of complex-valued functions. Hence, we shall leave the family $\mathbb{B}_I$ largely unconstrained except for two properties that are sufficient for the extension to work.
\begin{assumption}\label{assump:extendable_family}
 The family of functions $\mathbb{B}_I$ satisfy the following properties:
    \begin{enumerate}[label=\textnormal{(B\arabic*)}]
    \item \label{assump:extendable_family:bi-consistency} \textit{Bi-consistency}: for every $n\in\mathbb{N}$, $j=1,\dots,n$, and every $\bm{t}_n=(t_n,\dots,t_1)$, we have\small
    \begin{equation}
        \sum_{f^+_j,f^-_j\in\Omega}B_{\bm{t}_n}(\bm{f}^+_n,\bm{f}^-_n)=B_{t_1,\dots,\cancel{t_j},\dots,t_n}(f^+_n,f^-_n;\dots;\cancel{f^+_j,f^-_j};\dots;f^+_1,f^-_1).
    \end{equation}\normalsize
    \item \label{assump:extendable_family:uni_bound} \textit{Uniform boundedness}:
        \begin{equation}\label{eq:abstract_uni_bound}
            \sup\Big\{\|B_{\bm{t}_n}\|_1\,:\, B_{\bm{t}_n}\in\mathbb{B}_I\Big\} <\infty,\quad\text{where}\quad\|B_{\bm{t}_n}\|_1:=
             \sum_{\bm{f}^+_n,\bm{f}^-_n}\left|B_{\bm{t}_n}(\bm{f}^+_n,\bm{f}^-_n)\right|.
        \end{equation}\normalsize
    \end{enumerate}
\end{assumption}
Note that the uniform boundedness is trivially satisfied in the case of standard (that is, positive-valued) probability distributions. Indeed, if the functions satisfy $0\leq B_{\bm{t}_n}(\bm{f}_n^+,\bm{f}_n^-)\leq 1$, and are normalized, $\sum_{\bm{f}^\pm_n}B_{\bm{t}_n}(\bm{f}_n^+,\bm{f}_n^-)=1$, then by definition $\|B_{\bm{t}_n}\|_1 = 1$ for the whole family. For this reason, in the traditional statement of the extension theorem, the second assumption would be superfluous. However, for our purposes (cf.~Section~\ref{sec:quantum_extension}), we shall be dealing with the complex-valued case, whence the need for~\ref{assump:extendable_family:uni_bound}.

The idea at the root of the extension theorem is that each distribution in $\mathbb{B}_I$, each defining a discrete complex-valued measure on $\Omega^n\times\Omega^n$, can be unambiguously obtained from a unique, complex-valued ``master measure'' on the space $\Omega^I\times\Omega^I$ of \textit{pairs of trajectories}, i.e., the space of functions
 \begin{align}
         I\ni t \mapsto (f^+(t),f^-(t))\in \Omega\times\Omega
    \end{align}
Technically, complex-valued measures are set functions defined on a properly chosen class of sets---a $\sigma$-algebra. There is indeed a standard procedure for constructing a $\sigma$-algebra on spaces of trajectories, cf.~Definition~\ref{def:product_space} in the Appendix~\ref{app:measure_theory} collecting all technical details; this construction shall be tacitly assumed in the following.
\begin{theorem}[Extension theorem for bi-probabilities]\label{thm:extension_theorem}
    Let $\mathbb{B}_I$ satisfy Assumptions~\ref{assump:extendable_family:bi-consistency} and~\ref{assump:extendable_family:uni_bound}. Then there exists a unique complex-valued measure $\mathcal{B}_I[f^+,f^-][\mathcal{D}f^+][\mathcal{D}f^-]$ on the space of trajectory pairs that uniquely extends every member of $\mathbb{B}_I$ in the following sense: $B_{\bm{t}_n}(\bm{f}^+_n,\bm{f}^-_n)$ equals the integral of $\mathcal{B}_I[f^+,f^-][\mathcal{D}f^+][\mathcal{D}f^-]$ on the set of all pairs of trajectories satisfying $f^\pm(t_j)=f^\pm_j$, $j=1,\dots,n$, i.e.,
    \begin{align}\label{eq:gen_extension_thrm:extension_formula}
        B_{\bm{t}_n}(\bm{f}^+_n,\bm{f}_n^-) = \iint \Big(\prod_{j=1}^n
            \delta_{f^+(t_j),f_j^+}\,\delta_{f^-(t_j),f_j^-}
        \Big)\mathcal{B}_I[f^+,f^-][\mathcal{D}f^+][\mathcal{D}f^-].
    \end{align}
\end{theorem}
This theorem is a special case (see Remark~\ref{rem:pro}) of a more general statement, Theorem~\ref{thm:extension_theorem_pro}, that essentially generalizes the traditional Kolmogorov extension theorem to include complex-valued measures. The full proof of this generalization, together with an overview of all the needed mathematical theory, is presented in Appendix~\ref{app:measure_theory}. Referring to the appendix for all due details, below we sketch a simplified formulation of the proof, adapted to the case examined here. This should suffice in illustrating the ideas behind it and the necessity of \textit{both} assumptions~\ref{assump:extendable_family:bi-consistency} and~\ref{assump:extendable_family:uni_bound}.

\begin{proof}[Sketch of the proof]

In plain words, the theorem asserts that, for a bi-consistent and uniformly bounded family of multi-variable functions, it is possible to take the limit of infinitely many arguments. However, instead of performing the limit on the functions $B_{\bm{t}_n}$ themselves, the idea is to carry it out at the level of the corresponding \textit{averages}---that is, the linear functionals defined by
\begin{align}
    \operatorname{E}_{\bm{t}_n}[X_{\bm{t}_n}]:= \sum_{\bm{f}_n^\pm}B_{\bm{t}_n}(\bm{f}_n^+,\bm{f}_n^-)X_{\bm{t}_n}(\bm{f}_n^+,\bm{f}_n^-)\label{eq:average}
\end{align}
with $X_{\bm{t}_n}:\,\Omega^n\times\Omega^n \to \mathbb{C}$. The proof then boils down to demonstrating the existence of a unique ``master average'', that is, a functional $\mathcal{E}_I$ satisfying the following property: given any continuous $\mathcal{X}_I:\Omega^I\times\Omega^I\to\mathbb{C}$, we have
\begin{align}\label{eq:gen_extension_thrm:limit_formal_def}
    \mathcal{E}_I[\mathcal{X}_I] = \lim_{\bm{t}_n\to I}\operatorname{E}_{\bm{t}_n}[X_{\bm{t}_n}],
\end{align}
for any sequence of multi-variable functions $(X_{\bm{t}_n})_{\bm{t}_n}$ such that $\|X_{\bm{t}_n} - \mathcal{X}_I\|_\infty \to 0$ as $\bm{t}_n \to I$, where $\|\cdot\|_\infty$ denotes the supremum norm. The Stone--Weierstrass theorem ensures that such sequence always exists for any continuous $\mathcal{X}_I$. Here, $\bm{t}_n \to I$ is used as a shorthand for the corresponding set limit, $\lim_{n\to\infty} I_n\cap I = I$, with
\begin{align}
    I_n := \{ t\in I\,:\,\bm{t}_n = (t_n,\ldots,t_1),\  t=t_n,\ldots,t_1\},
\end{align}
where the sequence $(\bm{t}_n)_{n=1}^\infty$ is chosen in such a way that $I_1\subset I_2 \subset\cdots\subset I$. 

We must therefore show that the convergence of $(X_{\bm{t}_n})_{\bm{t}_n}$ to a given $\mathcal{X}_I$ implies the convergence of $(\operatorname{E}_{\bm{t}_n}[X_{\bm{t}_n}])_{\bm{t}_n}$, \textit{independently} of the choice of the sequence of functions. To this end it is sufficient and necessary to verify that $(\operatorname{E}_{\bm{t}_n}[X_{\bm{t}_n}])_{\bm{t}_n}$ is a Cauchy sequence, which in this case follows from~\ref{assump:extendable_family:bi-consistency} and~\ref{assump:extendable_family:uni_bound},
\begin{align}\label{eq:extension}
\nonumber
    \big|\operatorname{E}_{\bm{t}_n}[X_{\bm{t}_n}]-\operatorname{E}_{\bm{t}_m}[X_{\bm{t}_m}]\big| &=
        \big|\operatorname{E}_{\bm{t}_n\cup\bm{t}_m}[X_{\bm{t}_n}]-\operatorname{E}_{\bm{t}_n\cup\bm{t}_m}[X_{\bm{t}_m}]\big|
        =\big|\operatorname{E}_{\bm{t}_n\cup\bm{t}_m}[X_{\bm{t}_n}-X_{\bm{t}_m}]\big|\\
\nonumber
    &\leq \|X_{\bm{t}_n}-X_{\bm{t}_m}\|_\infty \|B_{\bm{t}_n\cup\bm{t}_m}\|_1\\
    &\leq \Big(\sup_{B_{\bm{t}_k}\in\mathbb{B}_I}\|B_{\bm{t}_k}\|_1\Big)\|X_{\bm{t}_n}-X_{\bm{t}_m}\|_\infty
    \xrightarrow{\bm{t}_n,\bm{t}_m\to I} 0,
\end{align}
with $\bm{t}_n\cup\bm{t}_m$ indicating $N$-tuple $\bm{t}_N$ consisting of elements from both $\bm{t}_n$ and $\bm{t}_m$ (and thus, corresponding to set $I_N = I_n\cup I_m$). The first two equalities in~\eqref{eq:extension}, where the linearity of $\operatorname{E}$ is used on two functions with \textit{unequal} tuples $\bm{t}_n\neq \bm{t}_m$, are enabled by bi-consistency~\ref{assump:extendable_family:bi-consistency}: indeed, if we take e.g. $\bm{t}_{n}=(t_n,\ldots,t_1)$ and $\bm{t}_m = (t_n,\ldots,\cancel{t_j},\ldots,t_1)$, then
\begin{align}
\nonumber
    \operatorname{E}_{\bm{t}_m}[X_{\bm{t}_m}] &= \sum_{\bm{f}^\pm_m} B_{\bm{t}_m}(\bm{f}_m^+,\bm{f}_m^-)X_{\bm{t}_m}(\bm{f}_m^+,\bm{f}_m^-)
        = \sum_{\bm{f}_m^\pm}\Big(\sum_{f_j^\pm}B_{\bm{t}_n}(\bm{f}_n^+,\bm{f}_n^-)\Big)X_{\bm{t}_m}(\bm{f}_m^+,\bm{f}_m^-)\\
    &=\sum_{\bm{f}_n^\pm}B_{\bm{t}_n}(\bm{f}_n^+,\bm{f}_n^-)X_{\bm{t}_m}(\bm{f}_m^+,\bm{f}_m^-)
        = \operatorname{E}_{\bm{t}_n}[X_{\bm{t}_m}] = \operatorname{E}_{\bm{t}_n\cup\bm{t}_m}[X_{\bm{t}_m}].
\end{align}
Of course, the last inequality in~\eqref{eq:extension} is due to the uniform boundedness~\ref{assump:extendable_family:uni_bound} of $\mathbb{B}_I$.

Thus, the existence of a functional $\mathcal{E}_I$ satisfying Eq.~\eqref{eq:gen_extension_thrm:limit_formal_def} has been proven. Since measures and the averages are in a one-to-one correspondence (Theorem~\ref{thm:riesz-markov}), this ultimately guarantees the existence of a measure on trajectories satisfying
\begin{align}
    \iint \mathcal{X}_I[f^+,f^-]\,\mathcal{B}_I[f^+,f^-][\mathcal{D}f^+][\mathcal{D}f^-] &= \mathcal{E}_I[\mathcal{X}_I]
\end{align}
for every continuous $\mathcal{X}_I:\,\Omega^I\times\Omega^I\to\mathbb{C}$. Finally, to prove Eq.~\eqref{eq:gen_extension_thrm:extension_formula}, take $\bm{t}_n\in I^n$ and $\bm{f}^\pm_n\in\Omega^n$, let $\chi_{\bm{t}_n}$ be the corresponding characteristic function:
\begin{equation}
    \chi_{\bm{t}_n}(\bm{\phi}_n^+,\bm{\phi}_n^-)=\delta_{\bm{\phi}^+_n,\bm{f}_n^+}\delta_{\bm{\phi}_n^-,\bm{f}_n^-}
        =\prod_{j=1}^n\delta_{\phi^+_j,f^+_j}\delta_{\phi^-_j,f^-_j}
\end{equation}
and define $\chi_{I}$ by
\begin{equation}
    \chi_{I}[f^+,f^-]:=\chi_{\bm{t}_n}\left(f^+(t_n),f^-(t_n);\ldots;f^+(t_1),f^-(t_1)\right)
        =\prod_{j=1}^n\delta_{f^+(t_j),f^+_j}\delta_{f^-(t_j),f^-_j}.
\end{equation}
Then by construction $\mathrm{E}_{\bm{t}_n}[\chi_{\bm{t}_n}]=\mathcal{E}_I[\chi_{I}]$, and a direct check shows that the former and the latter respectively correspond to the left- and right-hand sides of Eq.~\eqref{eq:gen_extension_thrm:extension_formula}.

\end{proof}

\section{The extension theorem in quantum mechanics}\label{sec:quantum_extension}

Let us now consider a quantum system represented by a finite-dimensional Hilbert space $\mathcal{H}$, the dynamics of which is described by the unitary evolution operator ($\hbar=1$):
\begin{align}
    \hat U_{t,t'}&=\mathcal{T}\e^{-\i\int_{t'}^t\hat H(s)ds}
        = \sum_{k=0}^\infty (-\i)^k\int_{t'}^t\mathrm{d}s_k\int_{t'}^{s_{k}}\mathrm{d}s_{k-1}\cdots\int_{t'}^{s_{2}}\mathrm{d}s_1\prod_{j=k}^1\hat H(s_j),
\end{align}
with $\hat{H}(t)$ being the (possibly time-dependent) Hamiltonian of the system.
The initial state of the system is represented by a density operator $\hat{\rho}$ on $\mathcal{H}$, and it is considered here to be fixed.

In this setting, one constructs the family of \textit{quantum} bi-probabilities associated with the time interval $I\subseteq\mathbb{R}_+$ and an observable $F$ represented by a discrete-valued projection-valued measure (PVM) $\{\hat{P}(f)\}_{f\in\Omega}$, with $\Omega$ being a set of finite cardinality,
\begin{align}
    \mathbb{Q}_{I}^F := \{ Q_{\bm{t}_n}(\bm{f}_n^+,\bm{f}^-_n)\ :\ n\in\mathbb{N},\,0=t_0<t_1<\ldots<t_n,\,t_1,\ldots,t_n\in I\}
\end{align}
consisting of functions defined according to
\begin{equation}\label{eq:quantum_bi-prob}
    Q_{\bm{t}_n}(\bm{f}^+_n,\bm{f}^-_n):=\operatorname{tr}\Big[\Big(\prod_{j=n}^1\hat{P}_{t_j}(f^+_j)\Big)\hat{\rho}
    \,\Big(\prod_{j=1}^n\hat{P}_{t_j}(f^-_j)\Big)\Big],
\end{equation}
with $\hat{P}_t(f)=\hat U_{0,t}\hat{P}(f)\hat U_{t,0}$. 

Bi-probabilities find their use in multiple contexts of standard quantum theory as a flexible parameterization of the dynamics of the observable they are associated with. Examples include the descriptions of the quantum--classical transition~\cite{Szankowski_2024} and, as already mentioned in Section~\ref{subsec:main_result}, of open system dynamics~\cite{Szankowski_SciRep20,Szankowski_PRA21,Szankowski_SciPostLecNotes23}. Of course, the family of bi-probabilities also contains the description of sequential measurements of the observable $F$~\cite{Szankowski_2024,Szankowski_PRA21}. Accordingly, together with $\mathbb{Q}_I^F$, we shall also consider the family of probabilities describing sequences of projective measurements of $F$ performed on the system,
\begin{align}\label{eq:quantum_prob}
    \mathbb{P}_I^F := \{ P_{\bm{t}_n}(\bm{f}_n)\ :\ n\in \mathbb{N},\ 0=t_0<t_1<\ldots<t_n,\ t_1,\ldots,t_n\in I\},
\end{align}
where each member function is defined as the diagonal part of the corresponding quantum bi-probability:
\begin{equation}
    P_{\bm{t}_n}(\bm{f}_n):=Q_{\bm{t}_n}(\bm{f}_n,\bm{f}_n)
        = \operatorname{tr}\Big[\Big(\prod_{j=n}^1\hat P_{t_j}(f_j)\Big)\hat\rho\,\Big(\prod_{j=1}^n\hat P_{t_j}(f_j)\Big)\Big].
\end{equation}

\begin{theorem}[Extension theorem for quantum bi-probabilities]\label{thm:quantum_extension}
    Let $\mathcal{H}$ be a Hilbert space with dimension $\dim\mathcal{H}=d<\infty$, $\hat{\rho}$ a density operator on $\mathcal{H}$, $\{\hat{P}(f)\}_{f\in\Omega}$ a projection-valued measure, and $\hat{H}(t)$ a continuously time-dependent Hamiltonian over a compact interval $I\subset\mathbb{R}$. Let $\mathbb{Q}^F_I$ the corresponding family of complex bi-probability distributions as defined in Eq.~\eqref{eq:quantum_bi-prob}. 
    Then there exists a complex-valued measure $\mathcal{Q}_I[f^+,f^-]\,[\mathcal{D}f^+][\mathcal{D}f^-]$ on the space of trajectory pairs that uniquely extends $\mathbb{Q}^F_I$ in the following sense:
    \begin{equation}
        Q_{\bm{t}_n}(\bm{f}^+_n,\bm{f}^-_n)=\iint \Big(\prod_{j=1}^n 
                \delta_{f^+(t_j),f^+_j}\,\delta_{f^-(t_j),f^-_j}
        \Big)\mathcal{Q}_I[f^+,f^-][\mathcal{D}f^+][\mathcal{D}f^-],
    \end{equation}
and, in particular,
\begin{equation}
        P_{\bm{t}_n}(\bm{f}_n)=\iint \Big(\prod_{j=1}^n 
            \delta_{f^+(t_j),f_j}\,\delta_{f^-(t_j),f_j}
        \Big)\mathcal{Q}_I[f^+,f^-][\mathcal{D}f^+][\mathcal{D}f^-].
    \end{equation}
\end{theorem}
This result follows from Theorem~\ref{thm:extension_theorem}, provided that the family $\mathbb{Q}_I^F$ is bi-consistent~\ref{assump:extendable_family:bi-consistency} and uniformly bounded~\ref{assump:extendable_family:uni_bound}. Accordingly, the remainder of this section is devoted to showing that $\mathbb{Q}^F_I$ does indeed satisfy both of these properties; while bi-consistency will simply follow by an immediate computation, uniform boundedness will be proven as a non-trivial consequence of the physical properties of the dynamics in quantum systems. We postpone the discussion on this result to Section~\ref{sec:discussion}.

\subsection{Properties of quantum bi-probabilities}

We shall start by listing some basic properties of the families $\mathbb{Q}^F_I$ and $\mathbb{P}^F_I$ that follow directly from the definition~\eqref{eq:quantum_bi-prob}.
\begin{proposition}\label{prop:properties_quantum_bi-probs}
 The family of quantum bi-probabilities $\mathbb{Q}_I^F$ as in Eq.~\eqref{eq:quantum_bi-prob} has the following properties:
    \begin{enumerate}[label=\textnormal{(Q\arabic*)}]
    \item \label{prop:normalization} normalization:
    \begin{equation*}
        \sum_{\bm{f}_n^+,\bm{f}_n^-\in\Omega^n}Q_{\bm{t}_n}(\bm{f}_n^+,\bm{f}_n^-)=1.
    \end{equation*}
    \item \label{prop:causality} causality:
    \begin{equation*}
        Q_{t_n,\ldots,t_1}(f^+_n,f^-_n;\dots;f^+_1,f^-_1)=\delta_{f^+_n,f^-_n}Q_{t_n,\dots,t_1}(f^+_n,f^-_n;\dots;f^+_1,f^-_1).
    \end{equation*}
    \item \label{prop:positivity} positive semidefiniteness: given any function $\Omega^n\ni \bm{f}_n\mapsto Z(\bm{f}_n)\in\mathbb{C}$,
    \begin{equation*}
        \sum_{\bm{f}_n^+,\bm{f}_n^-}Z(\bm{f}_n^+)Q_{\bm{t}_n}(\bm{f}_n^+,\bm{f}_n^-)Z(\bm{f}_n^-)^* \geq 0.
    \end{equation*}
    \item \label{prop:bi-consistency} \textbf{bi-consistency}:
    \begin{equation*}
        \sum_{f^+_j,f^-_j\in\Omega}Q_{\bm{t}_n}(\bm{f}_n^+,\bm{f}_n^-)=Q_{t_n,\ldots,\cancel{t_j},\ldots,t_1}(f^+_n,f^-_n;\ldots;\cancel{f^+_j,f^-_j};\ldots;f^+_1,f^-_1).
    \end{equation*}\normalsize
\end{enumerate}
\end{proposition}
\begin{proof}
   Property~\ref{prop:normalization} is an immediate consequence of the normalization of the PVM, $\sum_{f}\hat P_t(f) = \hat 1$ for any $t\in I$;~\ref{prop:causality} follows from the cyclic property of the trace and orthogonality of projector operators, $\hat P_{t_n}(f_n^-)\hat P_{t_n}(f_n^+) \propto \delta_{f_n^+,f_n^-}$. Let us show~\ref{prop:positivity},
    \begin{align}     
    \nonumber
        &\sum_{\bm{f}_n^+,\bm{f}_n^-} Z(\bm{f}^+_n)Q_{\bm{t}_n}(\bm{f}^+_n,\bm{f}^-_n)Z(\bm{f}_n^-)^*\\
    \nonumber
        &\phantom{=}=
            \operatorname{tr}\left[\left(\sum_{\bm{f}_n^+}Z(\bm{f}_n^+)\Big(\prod_{j=n}^1\hat{P}_{t_j}(f^+_j)\Big)\hat\rho^{1/2}\right)
            \left(\sum_{\bm{f}_n^-}\hat\rho^{1/2}\Big(\prod_{j=1}^n\hat{P}_{t_j}(f^-_j)\Big)Z(\bm{f}_n^-)^*\right)\right]\\
    \nonumber
        &\phantom{=}=
            \operatorname{tr}\left[\left(\sum_{\bm{f}_n}Z(\bm{f}_n)\Big(\prod_{j=n}^1\hat{P}_{t_j}(f_j)\Big)\hat\rho^{1/2}\right)
            \left(\sum_{\bm{f}_n}Z(\bm{f}_n)\Big(\prod_{j=n}^1\hat{P}_{t_j}(f_j)\Big)\hat\rho^{1/2}\right)^\dagger\right]
        \geq 0.
    \end{align}
    Finally, the key bi-consistency~\ref{prop:bi-consistency} results from a direct calculation,
    \begin{align}
    \nonumber
        \sum_{f_j^+,f_j^-}Q_{\bm{t}_n}(\bm{f}_n^+,\bm{f}_n^-) &= \operatorname{tr}\Big[
            \hat P_{t_n}(f_n^+)\cdots\Big(\sum_{f_j^+}\hat P_{t_j}(f_j^+)\Big)\cdots\hat P_{t_1}(f_1^+)
            \hat\rho\\
    \nonumber
        &\phantom{=\operatorname{tr}\Big[}\times
            \hat P_{t_1}(f_1^-)\cdots\Big(\sum_{f_j^-}\hat P_{t_j}(f_j^-)\Big)\cdots\hat P_{t_n}(f_n^-)
        \Big]\\
    \nonumber
        &=\operatorname{tr}\Big[\hat P_{t_n}(f_n^+)\cdots\cancel{\hat P_{t_j}(f_j^+)}\cdots\hat P_{t_1}(f_1^+)\hat\rho
            \hat P_{t_1}(f_1^-)\cdots\cancel{\hat P_{t_j}(f_j^-)}\cdots\hat P_{t_n}(f_n^-)\Big]\\
        &= Q_{t_n,\ldots,\cancel{t_j},\ldots,t_1}(f_n^+,f_n^-;\ldots;\cancel{f_j^+,f_j^-};\ldots;f_1^+,f_1^-).
    \end{align}
\end{proof}

\begin{proposition}\label{prop:diagonal_part}
 The family of functions $\mathbb{P}_I^F$ as in Eq.~\eqref{eq:quantum_prob} has the following properties:
\begin{enumerate}[label=\textnormal{(P\arabic*)}]
    \item\label{prop:diag_prob} joint probability distributions:
    \begin{align*}
        P_{\bm{t}_n}(\bm{f}_n) \geq 0,\qquad\sum_{\bm{f}_n}P_{\bm{t}_n}(\bm{f}_n) = 1;
    \end{align*}
    \item\label{prop:CS} \textit{bi-probability bounding}:
    \begin{equation*}
        |Q_{\bm{t}_n}(\bm{f}_n^+,\bm{f}^-_n)|\leq \sqrt{P_{\bm{t}_n}(\bm{f}^+_n)}\,\sqrt{P_{\bm{t}_n}(\bm{f}^-_n)};
    \end{equation*}
    \item\label{prop:diag_causal} \textit{causality of measurements}:
    \begin{equation*}
         \sum_{f_{n+1}}P_{t_{n+1},t_{n},\dots,t_1}(f_{n+1};f_{n};\dots;f_1)=P_{\bm{t}_n}(\bm{f}_n);
    \end{equation*}
    \item\label{prop:inconsistency} \text{measurement inconsistency}:
        \begin{align*}
            P_{t_n,\ldots,\cancel{t_j},\ldots,t_1}(f_n,\ldots,\cancel{f_j},\ldots,f_1) - \sum_{f_j}P_{\bm{t}_n}(\bm{f}_n)
            = \sum_{f_j^+ \neq f_j^-} Q_{\bm{t}_n}(f_n,f_n;\ldots;f_j^+,f_j^-;\ldots;f_1,f_1).
        \end{align*}
\end{enumerate}
\end{proposition}

\begin{proof}
    The estimate~\ref{prop:CS} comes from the Cauchy--Schwarz inequality for the Frobenius inner product $(\hat A,\hat B) = \operatorname{tr}[\hat A^\dagger\hat B]$,
    \begin{align}
    \nonumber
        |Q_{\bm{t}_n}(\bm{f}_n^+,\bm{f}_n^-)|^2 &= \Big|\Big(
            \hat\rho^{1/2}\prod_{j=1}^n\hat P_{t_j}(f_j^+) \ ,\  \hat\rho^{1/2}\prod_{j=1}^n\hat P_{t_j}(f^-_j)
        \Big)\Big|^2\\
    \nonumber
        &\leq 
        \Big(\hat\rho^{1/2}\prod_{j=1}^n\hat P_{t_j}(f_j^+) \ ,\  \hat\rho^{1/2}\prod_{j=1}^n\hat P_{t_j}(f^+_j)\Big)
        \Big(\hat\rho^{1/2}\prod_{j=1}^n\hat P_{t_j}(f_j^-) \ ,\  \hat\rho^{1/2}\prod_{j=1}^n\hat P_{t_j}(f^-_j)\Big)\\
        &=Q_{\bm{t}_n}(\bm{f}_n^+,\bm{f}_n^+)Q_{\bm{t}_n}(\bm{f}_n^-,\bm{f}_n^-).
    \end{align}
    Property~\ref{prop:diag_causal} follows from the causality~\ref{prop:causality} and the bi-consistency~\ref{prop:bi-consistency} of bi-probabilities,
    \begin{align}
        \sum_{f_n}P_{\bm{t}_n}(\bm{f}_n)
            &=\sum_{f_n}Q_{\bm{t}_n}(\bm{f}_n,\bm{f}_n)
                = \sum_{f_n^+,f_n^-}Q_{\bm{t}_n}(f_n^+,f_n^-;\bm{f}_{n-1},\bm{f}_{n-1})
                = P_{\bm{t}_{n-1}}(\bm{f}_{n-1}),
    \end{align}
Then, the normalization of $P_{\bm{t}_n}$ is obtained by a repeated application of~\ref{prop:diag_causal} and the non-negativity follows from positive semidefiniteness~\ref{prop:positivity} with $Z(\bm{f}_n^\pm) = \prod_{j=1}^n \delta_{f_j^\pm,f_j}$. Finally,~\ref{prop:inconsistency} is simply the bi-consistency~\ref{prop:bi-consistency} for the diagonal part of bi-probabilities.
\end{proof}

The relation~\ref{prop:inconsistency} explains why the single-trajectory picture is inadequate for the description of quantum measurements~\cite{Szankowski_2024}: The remainder on the right-hand-side, that consists of a combination of ``off-diagonal'' parts of bi-probabilities, is the cause for violation of the classical Kolmogorov consistency~\eqref{eq:kolmogorov_consistency}. Consequently, the (classical) extension theorem for probabilities $\mathbb{P}_I^F$ does not work here, and without it the (single) trajectory interpretation for the observed measurements cannot be valid.

\subsection{Uniform bound for quantum bi-probabilities}

While bi-consistency is trivially satisfied for $\mathbb{Q}_I^F$, see property~\ref{prop:bi-consistency}, uniform boundedness,
\begin{align}
    \sup\big\{\|Q_{\bm{t}_n}\|_1\,:\, Q_{\bm{t}_n}\in\mathbb{Q}_I^F\big\} < \infty,
\end{align}
is not obvious. It would be tempting to guess that the uniform bound can be somehow derived from the normalization~\ref{prop:normalization} coupled with some of the other already identified properties of $\mathbb{Q}_I^F$. This turns out to not be the case, however. By only relying on the previous properties, the best one can do is the following (non-uniform) bound,
\begin{align}\label{eq:nonuniform_bound}
\|Q_{\bm{t}_n}\|_1&=\sum_{\bm{f}^+_n,\bm{f}^-_n}|Q_{\bm{t}_n}(\bm{f}^+_n,\bm{f}^-_n)|
    \leq\left(\sum_{\bm{f}_n}\sqrt{P_{\bm{t}_n}(\bm{f}_n)}\right)^2
    \leq\left(|\Omega|^{\frac{n}{2}}\sum_{\bm{f}_n}P_{\bm{t}_n}(\bm{f}_n)\right)^2= |\Omega|^{n},
\end{align}
where we used H\"older's inequality and properties~\ref{prop:diag_prob} and~\ref{prop:CS}, and $|\Omega|\leq d$ is the number of elements in $\Omega$. As this bound depends explicitly on the number of arguments $n$, it cannot be used for our purpose.

To show that the family of quantum bi-probabilities $\mathbb{Q}_I^F$ is uniformly bounded, one has to supplement its mathematical properties with two additional \textit{physical} constraints. The first one is to restrict bi-probabilities to a \textit{finite} time window, i.e., hereafter, $I=[0,T]$, $0<T<\infty$. The second necessary constraint---or rather a natural feature of quantum dynamics---is the strong continuity of quantum evolution. The significance of both will soon become apparent.

\begin{definition}[Time mesh refinement]\label{def:refinement}
    Let $n,N\in\mathbb{N}$ and $\bm{t}_n\in I^n$, $\bm{\tau}_N\in I^N$ be two ordered tuples of times. We say that $\bm{\tau}_N$ is a \textit{refinement} of $\bm{t}_n$ if the following conditions hold: $N\geq n$, and there exists a map $\iota:\{1,2,\dots,n\}\rightarrow\{1,2,\dots,N\}$ such that $\iota(n)=N$, and
    \begin{equation}
        \tau_{\iota(j)}=t_j,\qquad j=1,\dots,n.
    \end{equation}
\end{definition}    
    
\begin{lemma}\label{lemma:refinement}
    Let $n,N\in\mathbb{N}$ with $n\leq N$, $\bm{t}_n\in I^n$, and $\bm{\tau}_N\in I^N$ such that $\bm{\tau}_N$ is a \textit{refinement} of $\bm{t}_n$, then
    \begin{equation}
        \|Q_{\bm{t}_n}\|_1\leq \|Q_{\bm{\tau}_N}\|_1.
    \end{equation}
\end{lemma}
\begin{proof}
    Denote, for all $j=1,\dots,n$, $\phi^\pm_{\iota(j)}\equiv f^\pm_j$, where $\tau_{\iota(j)}=t_j$ as per Definition~\ref{def:refinement}. Then, using bi-consistency~\ref{prop:bi-consistency} we write
\begin{align}
\nonumber
    Q_{\bm{t}_n}(\bm{f}^+_n,\bm{f}^-_n)&=Q_{\tau_{\iota(n)},\dots,\tau_{\iota(1)}}(\bm{f}_n^+,\bm{f}_n^-)
    = Q_{\tau_{\iota(n)},\ldots,\tau_{\iota(1)}}(\phi^+_{\iota(n)},\phi^-_{\iota(n)};\ldots;\phi^+_{\iota(1)},\phi^-_{\iota(1)})\\
    &=\sum_{\phi^\pm_j:j\notin\iota(\{1,\dots,n\})}Q_{\bm{\tau}_N}(\bm{\phi}_N^+,\bm{\phi}_N^-),
\end{align}
which leads to the claimed property,
\begin{align}
\nonumber
    \|Q_{\bm{t}_n}\|_1 &=\sum_{\bm{f}^\pm_n}|Q_{\bm{t}_n}(\bm{f}_n^+,\bm{f}_n^-)|
    =\sum_{\phi_j^\pm:j\in\iota(\{1,\ldots,n\})}\bigg|\sum_{\phi^\pm_j:j\notin\iota(\{1,\ldots,n\})}
        Q_{\bm{\tau}_N}(\bm{\phi}_N^+,\bm{\phi}_N^-)\bigg|\\
\nonumber
    &\leq\sum_{\phi_j^\pm:j\in\iota(\{1,\ldots,n\})}\sum_{\phi^\pm_j:j\notin\iota(\{1,\ldots,n\})}
        |Q_{\bm{\tau}_N}(\bm{\phi}_N^+,\bm{\phi}_N^-)|
        =\sum_{\bm{\phi}_N^\pm}|Q_{\bm{\tau}_N}(\bm{\phi}_N^+,\bm{\phi}_N^-)|
        =\|Q_{\bm{\tau}_N}\|_1.
\end{align}
\end{proof}

\begin{proposition}\label{prop:quantum_uniform}
    Let $\mathbb{Q}^F_{[0,T]}$ be the family of quantum bi-probabilities associated with the time interval $I = [0,T]$ and the observable $F$ in the $d$-dimensional Hilbert space $\mathcal{H}$. Then
    \begin{equation}\label{eq:uniform_bound}
        \sup\left\{\|Q_{\bm{t}_n}\|_1\,:\, Q_{\bm{t}_n}\in\mathbb{Q}^F_{[0,T]}\right\}\leq d^2\,
                {\exp}\Big[2(d-1) \int_0^T\!\!\|\hat H(s)\|_\mathrm{op}\,\mathrm{d}s\Big]
            <\infty,
    \end{equation}
    where $\hat H(t)$ is the Hamiltonian of the system, and $\|\cdot\|_\mathrm{op}$ the operator norm. That is, the family is uniformly bounded in the sense of assumption~\ref{assump:extendable_family:uni_bound}.
\end{proposition}
\begin{proof}
To show that the whole family is bounded, we will proceed by first finding a uniform bound for its individual members. Without loss of generality, we will hereafter assume that the PVM $\{\hat P(f)\}_{f\in\Omega}$ representing the observable $F$ is composed entirely of rank-$1$ projectors, $\hat{P}_t(f)=\hat U_{0,t}|f\rangle\langle f|\hat U_{t,0}$, where $\{\ket{f}\}_{f\in\Omega}$ is a complete orthonormal basis of $\mathcal{H}$ and $|\Omega| = d$. This can be done without loss of generality. Indeed, if $\mathbb{Q}_I^F$ associated with such a non-degenerate observable is bounded, then the family associated with any degenerate observable $G$ represented by a PVM $\{\hat P(g)\}_{g\in\Omega'}$ such that $\hat P(g) = \sum_{f\in\omega(g)}|f\rangle\langle f|$,
\begin{align}
    \mathbb{Q}_I^{G} &= \Big\{ Q^G_{\bm{t}_n}(\bm{g}_n^+,\bm{g}_n^-) = \sum_{f_n^\pm\in\omega(g_n^\pm)}\cdots\sum_{f_1^\pm\in\omega(g_1^\pm)}Q_{\bm{t}_n}(\bm{f}_n^+,\bm{f}_n^-)\ :\  n\in\mathbb{N},\,t_1<\ldots<t_n,t_1,\ldots,t_n\in I\Big\},
\end{align}
will also be uniformly bounded,
\begin{align}
\nonumber
    \sup_{Q_{\bm{t}_n}^G\in\mathbb{Q}_I^G}\|Q_{\bm{t}_n}^G\|_1 &= \sup_{Q_{\bm{t}_n}^G\in\mathbb{Q}_I^G}
    \sum_{\bm{g}_n^\pm}\Big|\sum_{f_n^\pm\in\omega(g_n^\pm)}\cdots\sum_{f_1^\pm\in\omega(g_1^\pm)}Q_{\bm{t}_n}(\bm{f}_n^+,\bm{f}_n^-)\Big|\\
\nonumber
    &\leq \sup_{Q_{\bm{t}_n}\in\mathbb{Q}_I^F}
        \sum_{\bm{g}_n^\pm}\sum_{f_n^\pm\in\omega(g_n^\pm)}\cdots\sum_{f_1^\pm\in\omega(g^\pm_1)}|Q_{\bm{t}_n}(\bm{f}_n^+,\bm{f}_n^-)|\\
    &= \sup_{Q_{\bm{t}_n}\in\mathbb{Q}_I^F}\sum_{\bm{f}_n^\pm}|Q_{\bm{t}_n}(\bm{f}_n^+,\bm{f}_n^-)|
    = \sup_{Q_{\bm{t}_n}\in\mathbb{Q}_I^F}\|Q_{\bm{t}_n}\|_1.
\end{align}
Therefore, in what follows we will be dealing with quantum bi-probabilities for which Eq.~\eqref{eq:quantum_bi-prob} reads
\begin{equation}\label{eq:quantum_bi-prob_simple}
    Q_{\bm{t}_n}(\bm{f}^+_n,\bm{f}^-_n)=\delta_{f_n^+f_n^-}
    \sum_{f_0}\langle f_0^+|\hat\rho|f_0^-\rangle
    \prod_{j=0}^{n-1} 
        \langle f^+_{j+1}|\hat U_{t_{j+1},t_{j}}|f_{j}^+\rangle\,
        \langle f^-_{j+1}|\hat U_{t_{j+1},t_{j}}|f_{j}^-\rangle^*.
\end{equation}
In the next step, starting from a given time partition $\bm{t}_n = (t_n,\ldots,t_1)$ designating the chosen member bi-probability $Q_{\bm{t}_n}$, we define a sequence of its refinements $(\bm{\tau}_N)_{N=N_0}^\infty$. The initial $N_0>n$ is chosen in such a way that, for every $j=1,\ldots,n-1$, there exists $k_j\in\mathbb{N}$ for which
\begin{align}
    t_{j-1} < \frac{(k_j-1)}{N_0}t_n \leq t_j \leq \frac{k_j}{N_0}t_n < t_{j+1}.
\end{align}
This $N_0$ always exists because the rational numbers form a dense set of the real axis. Then, define the time partitioning scheme $\bm{\tau}_N$ for $N\geq N_0$ with the following procedure: (1) Assign a uniform partition,
\begin{align}
    \tau_k &= t_n k /N,\quad k=1,\ldots,N.
\end{align}
(2) Identify the times that are closest to the original partitioning,
\begin{align}
    K(N) = \{ k_j\ :\ t_n (k_j-1) /N \leq t_j \leq t_n k_j /N \}.
\end{align}
(3) ``Snap'' the identified times to the originals they approximate,
\begin{align}
    \tau_{k_j} = t_n k_j/N\, \to\, \tau_{k_j} = t_j,\quad k_j\in K(N),
\end{align}
which makes the resultant $\bm{\tau}_N$ a refinement of $\bm{t}_n$. The refinements defined this way are mostly uniform: $\Delta\tau_k = \tau_{k+1} -\tau_{k} = t_n/N$ for all $k\notin K(N)$, except for all $k\in K(N)$ such that
\begin{align}
    \Delta\tau_{k} = \frac{t_n}{N}+h_{k}\leq \frac{2t_n}{N},\quad \Delta\tau_{k-1} = \frac{t_n}{N} - h_{k} \leq \frac{t_n}{N},
    \quad 0 \leq h_k \leq \frac{t_n}{N}.
\end{align}
We then apply the refinements in a sequence according to Lemma~\ref{lemma:refinement},
\begin{align}
    \|Q_{\bm{t}_n}\|_1 \leq \| Q_{\bm{\tau}_{N_0}}\|_1 \leq \|Q_{\bm{\tau}_{2N_0}}\|_1 \leq \cdots
    \leq \|Q_{\bm{\tau}_{2^k N_0}}\|_1\leq \cdots \leq \lim_{N\to\infty}\|Q_{\bm{\tau}_N}\|_1.
\end{align}
In the process we have determined that $Q_{\bm{t}_n}$ is bounded by the limit of infinitely refined time mesh. Since for any pair of sequences $(a_N)_{N=N_0}^\infty$ and $(b_N)_{N=N_0}^\infty$ such that $a_N \leq b_N$ for all $N\geq N_0$ one has $\lim_{N\to\infty}a_N \leq \lim_{N\to\infty}b_N$, we are now allowed to estimate $\|Q_{\bm{\tau}_N}\|_1$ before taking the limit. To this end, we define the transition probabilities,
\begin{align}
    p^\text{stay}_{t,t'}(f) &:= |\langle f|\hat U_{t,t'}|f\rangle|^2,
    \qquad p^\text{leave}_{t,t'}(f) := 1 - p^\text{stay}_{t,t'}(f) =\sum_{f'\neq f}|\langle f'|\hat U_{t,t'}|f\rangle|^2.
\end{align}
The two functions satisfy $0\leq p^{\text{stay/leave}}_{t,t'}(f) \leq 1$, and their square roots can be estimated from above (using the Taylor theorem in the case of the ``leave'' probability),
\begin{subequations}
\begin{align}
    \delta_{f_{j+1}^\pm,f_{j}^\pm}
    |\langle f_{j+1}^{\pm}|\hat U_{\tau_{j+1},\tau_{j}}|f_{j}^\pm\rangle| &= \delta_{f_{j+1}^\pm,f_{j}^\pm}\sqrt{p_{\tau_{j+1},\tau_{j}}^\text{stay}(f_j^\pm)} \leq \delta_{f_{j+1}^\pm,f_{j}^\pm};\\
\nonumber
    (1\!-\!\delta_{f_{j+1}^\pm,f_{j}^\pm})|\langle f_{j+1}^\pm|\hat U_{\tau_{j+1},\tau_{j}}|f_{j}^\pm\rangle| 
    &\leq (1\!-\!\delta_{f_{j+1}^\pm,f_{j}^\pm})\sqrt{p^\text{leave}_{\tau_{j+1},\tau_{j}}(f_j^\pm)}\\
\nonumber
    &= (1\!-\!\delta_{f_{j+1}^\pm,f_{j}^\pm})\Big[
        \Delta\tau_j \sqrt{\langle f_j^\pm|\hat H(\tau_{j})^2|f_j^\pm\rangle -\langle f_j^\pm|\hat H(\tau_{j})|f_j^\pm\rangle^2}\\
\nonumber
    &\phantom{= (1\!-\!\delta_{f_j^\pm,f_{j-1}^\pm})}
        + R_{\tau_{j+1}}^{(1)}(f_j^\pm)\Big]\\
\nonumber
    &\leq (1\!-\!\delta_{f_{j+1}^\pm,f_{j}^\pm})\Big(\Delta\tau_j\|\hat H(\tau_{j})|f_j^\pm\rangle\|+\big|R_{\tau_{j+1}}^{(1)}(f_j^\pm)\big|\Big)\\
\nonumber
    &\leq (1\!-\!\delta_{f_{j+1}^\pm,f_{j}^\pm})\left(v(\tau_{j})\Delta\tau_j + \frac{1}{2}M\Delta\tau_j^2\right)\\
    &\leq (1\!-\!\delta_{f_{j+1}^\pm,f_{j}^\pm})\left(v(\tau_{j}) \Delta\tau_j + \frac{2M t_n^2}{N^2}\right),
\end{align}
\end{subequations}
where $0 < M <\infty$, $v(t) = \sup_{|f\rangle}\|\hat H(t)|f\rangle\| = \|\hat H(t)\|_\text{op}$, and $\Delta\tau_j^2 \leq 4 t_n^2/N^2$ which covers the possibility $j\in K(N)$ and $\Delta\tau_j = t_n/N+h_{j}$. These estimates lead to
\begin{align}
\nonumber
    \|Q_{\bm{\tau}_N}\|_1 &\leq \sum_{f_0^\pm}|\langle f_0^+|\hat\rho|f_0^-\rangle|
        \prod_{j=0}^{N-1}\sum_{f_{j+1}^+}\left[\delta_{f_{j+1}^+,f_{j}^+}+(1-\delta_{f_{j+1}^+,f_{j}^+})\left(v(\tau_{j})\Delta\tau_j + \frac{2Mt_n^2}{N^2}\right)\right]\\
\nonumber
    &\phantom{\sum_{f_0^\pm}|\langle f_0^+|\hat\rho|f_0^-\rangle|}\times
            \prod_{k=0}^{N-1}\sum_{f_{k+1}^-}\left[\delta_{f_{k+1}^-,f_{k}^-}+(1-\delta_{f_{k+1}^-,f_{k}^-})\left(v(\tau_{k})\Delta\tau_k +\frac{2Mt_n^2}{N^2}\right)\right]\\
\nonumber
    &= \left(\sum_{f_0^\pm}|\langle f_0^+|\hat\rho|f_0^-\rangle|\right)
        \prod_{j=0}^{N-1}\left[1+(d-1)\left(v(\tau_{j})\Delta\tau_j + \frac{2Mt_n^2}{N^2}\right)\right]^2\\
\nonumber
    &\leq  d^2 \prod_{j=0}^{N-1}\left[1+(d-1)\left(v(\tau_{j})\Delta\tau_j + \frac{2Mt_n^2}{N^2}\right)\right]^2\\
    &\leq d^2 \left[1+(d-1)\frac{2Mt_n^2}{N^2}\right]^{2N}
        \prod_{j=0}^{N-1}\left[1+(d-1)v(\tau_{j})\Delta\tau_j\right]^2.
\end{align}
In this form, the limit $N\to\infty$, which simultaneously implies $\Delta\tau_j \to 0$, results in
\begin{align}
    \|Q_{\bm{t}_n}\|_1 & \leq \lim_{N\to\infty}\|Q_{\bm{\tau}_N}\|_1
    \leq d^2
        {\exp}\Bigg[2(d-1)\lim_{\substack{N\to\infty;\\\Delta\tau_j\to 0}}\sum_{j=0}^{N-1}v(\tau_{j})\Delta\tau_j\Bigg]
        = d^2\,\e^{2(d-1)\int_0^{t_n}v(s)\mathrm{d}s}.
\end{align}
Since the above bound for individual bi-probabilities $Q_{t_n,\ldots,t_1}$ depends only on the latest time $t_n$, the bound for the whole family is identified without further difficulty,
\begin{align}
   \sup\left\{\|Q_{\bm{t}_n}\|_1\,:\, Q_{\bm{t}_n}\in\mathbb{Q}_{[0,T]}^F\right\} &\leq d^2\, 
        {\exp}\Big[2(d-1)\int_0^T\!\!\|\hat H(s)\|_\text{op}\,\mathrm{d}s\Big]
    < \infty,
\end{align}
so that the claim is proven.
\end{proof}
We have therefore shown that the family of quantum bi-probabilities introduced in Eq.~\eqref{eq:quantum_bi-prob} satisfies both assumptions~\ref{assump:extendable_family:bi-consistency} (by Proposition~\ref{prop:properties_quantum_bi-probs}) and~\ref{assump:extendable_family:uni_bound} (by Proposition~\ref{prop:quantum_uniform}) of the abstract extension result proven in Section~\ref{sec:gen_extension_thrm}, i.e.~Theorem~\ref{thm:extension_theorem}; this concludes the proof of Theorem~\ref{thm:quantum_extension}.

\section{Discussion}\label{sec:discussion}

\subsection{The extension theorem and the choice of time interval}\label{sec:the_matter_of_time}

We have shown that, in quantum mechanics, the extension theorem works for families of bi-probabilities confined to a specific time interval $I=[0,T]$. Note, however, that the members of the family associated with one interval also belong, in general, to other families associated with different intervals; in particular, if a given bi-probability $Q_{\bm{t}_n}$ is a member of $\mathbb{Q}_{[0,T]}^F$, then so is $Q_{\bm{t}_n}\in\mathbb{Q}_{[0,T']}^F$ when $T<T'$. Therefore, this $Q_{\bm{t}_n}$ is a \textit{restriction} of the master bi-measure $\mathcal{Q}_{[0,T]}$ that extends $\mathbb{Q}_{[0,T]}^F$, and simultaneously, of $\mathcal{Q}_{[0,T']}$ that extends $\mathbb{Q}_{[0,T']}^F$,
\begin{align}
\nonumber
    Q_{\bm{t}_n}(\bm{f}_n^+,\bm{f}_n^-) &=\iint \Big(\prod_{j=1}^n
            \delta_{f^+(t_j),f_j^+}\,\delta_{f^-(t_j),f_j^-}
        \Big)\mathcal{Q}_{[0,T]}[f^+,f^-][\mathcal{D}f^+][\mathcal{D}f^-]\\
        &= \iint \Big(\prod_{j=1}^n
            \delta_{f^+(t_j),f_j^+}\,\delta_{f^-(t_j),f_j^-}
        \Big)\mathcal{Q}_{[0,T']}[f^+,f^-][\mathcal{D}f^+][\mathcal{D}f^-].
\end{align}
This begs the question: what is the relation between the two master bi-measures $\mathcal{Q}_{[0,T]}$ and $\mathcal{Q}_{[0,T']}$?

We find the answer by considering the relations between the corresponding families of bi-probabilities. Observe that the bi-consistent family associated with a shorter interval is a \textit{subset} of the family corresponding to the longer interval, $\mathbb{Q}^F_{[0,T]}\subseteq \mathbb{Q}^F_{[0,T']}$. Moreover, the union of both families is a bi-consistent family itself and it is simply equal to the larger family,
\begin{align}
    \nonumber
    \mathbb{Q}_{[0,T]}^F\cup\mathbb{Q}_{[0,T']}^F &= \{ Q_{\bm{\tau}_k}(\bm{f}_k^+,\bm{f}^-_k)\ :\ 
        \forall_{k\geqslant 1}\ T\geqslant \tau_k \geqslant \cdots \geqslant \tau_1 \geqslant 0\,\lor\, T' \geqslant \tau_k \geqslant \cdots\geqslant \tau_1 \geqslant 0\}\\
    &=\{Q_{\boldsymbol{\tau}_k}(\bm{f}_k^+,\bm{f}^-_k)\ :\ 
        \forall_{k\geqslant 1}\ T'\geqslant \tau_k \geqslant \cdots \geqslant \tau_1 \geqslant 0\} = \mathbb{Q}^F_{[0,T']},
\end{align}
implying that the master bi-measure $\mathcal{Q}_{[0,T']}$ is an extension to $\mathbb{Q}_{[0,T']}^F$ as well as $\mathbb{Q}_{[0,T]}^F$. Therefore, given that the extension of a family is unique, it follows that $\mathcal{Q}_{[0,T]}$ has to be the \textit{restriction} of $\mathcal{Q}_{[0,T']}$.

This reasoning can be iterated any number of times; given a sequence of increasing time intervals, $I_1 \subset I_2 \subset \cdots \subset I_N$, with $I_n = [0,T_n]$ and $0<T_1<\cdots<T_N$, we have an increasing sequence of bi-consistent families, $\mathbb{Q}^F_{I_1}\subseteq \cdots \subseteq \mathbb{Q}^F_{I_N}$, such that $\bigcup_{n=1}^N\mathbb{Q}^F_{I_n} = \mathbb{Q}^F_{I_N}$. As long as the largest time $T_N$ is a finite number, the bi-consistent families are uniformly bounded and they extend to the same master bi-measure $\mathcal{Q}_{I_N}$ (modulo an appropriate restriction to the time interval corresponding to a given family). Therefore, since $N$ can be arbitrary large---and thus, so can $T_N$---we can state that a family $\mathbb{Q}^F = \lim_{N\to\infty}\bigcup_{n=1}^N \mathbb{Q}^F_{I_n}$ extends to the master bi-measure $\mathcal{Q}$ such that any $\mathcal{Q}_{[0,T]}$ with a finite $T$ is its restriction. 

Note that this result cannot be interpreted to mean that $\mathbb{Q}^F = \mathbb{Q}^F_{[0,\infty)}$ or $\mathcal{Q} = \mathcal{Q}_{[0,\infty)}$. Indeed, a bi-consistent family associated with an interval $I$ is extendable only if it is uniformly bounded, which is not the case when $I$ is not compact, e.g., for $I=[0,\infty)$ the bound $d^2{\exp}[2(d-1)\int_0^\infty\|\hat H(s)\|_\mathrm{op}\mathrm{d}s]$ generally diverges. 
In fact, the inability to bound, and thus to extend, quantum bi-probabilities on $[0,\infty)$ in the general case could be regarded as a desirable feature rather than some limitation of the theory. If the master bi-probability $\mathcal{Q}_{[0,\infty)}$ would exist for any quantum system, then its single-time restriction,
\begin{align}
    Q_\infty(f,f') &= \delta_{f,f'}P_\infty(f) = \delta_{f,f'}\lim_{T\to\infty}\operatorname{tr}\big[\hat P_T(f)\hat\rho\big],
\end{align}
would also be sensible. However, a simple physical analysis shows that such a thing cannot be the case. We can illustrate this point with a basic example. Consider a $2$-dimensional system (qubit) with $\hat F = \hat\sigma_z$, $\hat H = \omega\hat\sigma_x/2$ and $\hat\rho = |{+1}\rangle\langle{+1}|$: then, for as long as time $T$ is a finite number (no matter how large or small it is), one gets the probability that has a well-defined physical and mathematical meaning,
\begin{align}
    Q_T(\pm 1,\pm 1) &= \frac{1\pm\cos\omega T}{2};
\end{align}
however, of course, this ceases to be the case when $T\to\infty$, as $\lim_{T\to\infty}Q_T$ does not exist. To make an infinite time a viable consideration, one would have to postulate some additional physical constraints on the system.

To sum up, as far as the extension theorem goes, a subtle but fundamental distinction between ``arbitrarily large time'' and ``infinite time'' holds. On the other hand, this implies that every property of quantum system which is formulated in terms of those of the underlying master measure only makes sense on finite time windows. For example, if one decides to deem a quantum system ``classical'' depending on whether the corresponding master measure $\mathcal{Q}_I[f^+,f^-]$ reduces to an actual, positive-valued measure on \textit{single} trajectories (i.e., $\mathcal{Q}_I[f,f']=\delta[f-f']\,\mathcal{P}_I[f]$), then one can only speak about classicality in finite times. Broadly speaking, this is reminiscent of the fact that quantum non-Markovianity, in its diverse declinations, is undecidable unless a finite time interval is selected \textit{a priori}.~\cite{burgarth2021hidden,burgarth2021quantum}

\subsection{Quantum bi-probabilities associated with multiple non-commuting observables}\label{subsec:multiple_observables}

Up to this point we have considered quantum bi-probabilities associated with a single observable. As stated previously, $\mathbb{Q}_I^F$ suffices for the purpose of employing the bi-probability parameterization in contexts involving only $F$ (and observables that commute with it), including the sequential measurements of $F$. A natural question is whether the bi-probability formalism can be generalized to handle \textit{multiple} non-commuting observables as well, so that it could parameterize, among others, multi-observable measurements. 

Let $\{F\}$ be a collection of observables, each represented by a corresponding PVM, 
\begin{align}
    \big\{\hat P^F(f)\big\}_{f\in\Omega(F)}\quad\text{where $F\in\{F\}$ and $|\Omega(F)|\leq d$.}
\end{align} 
Then define the family associated with time interval $I$ and the collection $\{F\}$,
\begin{align}
    \mathbb{Q}_{I|\rho}^{\{F\}} &= \big\{
        Q^{\bm{F}_n}_{\bm{t}_n}(\bm{f}_n^+,\bm{f}_n^-|\rho)\ :\ n\in\mathbb{N},\,0=t_0<t_1<\cdots<t_n,\, \forall_{j=1,\ldots,n}\ t_j\in I,\,F_{j}\in\{F\}
    \big\},
\end{align}
with its members given by
\begin{align}\label{eq:multi-obs_bi-prob}
    Q_{\bm{t}_n}^{\bm{F}_n}(\bm{f}_n^+,\bm{f}_n^-|\rho) := \operatorname{tr}\Big[
        \Big(\prod_{j=n}^1 \hat P^{F_{j}}_{t_j}(f_j^+)\Big)\hat\rho\Big(\prod_{j=1}^n \hat P^{F_{j}}_{t_j}(f_j^-)\Big)
    \Big].
\end{align}
These functions can be considered as bi-probability distributions, as it is straightforward to check that they share all the properties~\ref{prop:properties_quantum_bi-probs} of the single-observable quantum bi-probabilities; crucially, this includes the bi-consistency of $\mathbb{Q}_{I|\rho}^{\{F\}}$. Moreover, the family $\mathbb{P}_{I|\rho}^{\{F\}}$ composed of the diagonal parts of $Q_{\bm{t}_n}^{\bm{F}_n}$,
\begin{align}
    P_{\bm{t}_n}^{\bm{F}_n}(\bm{f}_n|\rho) &:= Q_{\bm{t}_n}^{\bm{F}_n}(\bm{f}_n,\bm{f}_n|\rho),
\end{align}
has the properties~\ref{prop:diagonal_part}, and thus, it constitutes the standard quantum mechanical description for the sequential projective measurements of \textit{multiple} observables. Finally, it can also be shown that the bi-probabilities belonging to $\mathbb{Q}_{I|\rho}^{\{F\}}$ are sufficient for parameterizing all other contexts involving observables $\{F\}$~\cite{Szankowski_SciRep20,Szankowski_SciPostLecNotes23}. Therefore, the simple answer to the question from before is that, indeed, the bi-probability formalism is not restricted to the case of single observable.

However, even though $\mathbb{Q}_{I|\rho}^{\{F\}}$ is a bi-consistent family, the matter of its extension to a multi-observable master bi-measure is not fully resolved at this moment---the issue of uniform boundedness turns out to be problematic. The method used in the single-observable case cannot be directly applied to the multi-observable one, at least not without introducing some additional structure on top of $\mathbb{Q}^{\{F\}}_{I|\rho}$. To better explain why this is the case, we first introduce the collection of unitary basis transformations $\{\hat U_F\}$ corresponding to each observable in $\{F\}$. That is, given the eigenbasis of $F\in\{F\}$, $\{ |f_k\rangle \}_{k=1}^d$, the operators of the PVM representing this observable can be written in the following form:
\begin{align}
    \hat P_t^{F}(f) &= \sum_{k=1}^d\delta\big(\langle f_k|\hat F|f_k\rangle - f\big)\,\hat U_{0,t}|f_k\rangle\langle f_k|\hat U_{t,0}
        = \sum_{k=1}^d\delta\big(\langle k|\hat U_F^\dagger \hat F\hat U_F|k\rangle - f\big)\hat U_{0,t}\hat U_F|k\rangle\langle k|\hat U_{F}^\dagger\hat U_{t,0},
\end{align}
where the unitary operator $\hat U_F$ transforms the chosen reference basis $\{|k\rangle\}_{k=1}^d$ into $\{|f_k\rangle\}_{k=1}^d$. Then, making use of our previous methods whenever possible, we arrive at the following estimation for the norm of a multi-observable family of bi-probabilities:
\begin{align}
\nonumber
    \|Q_{\bm{t}_n}^{\bm{F}_n}\|_1 &\leq \left(\sum_{k_n,\ldots,k_0=1}^d\sqrt{\rho(k_0)}\prod_{j=1}^n 
        |\langle k_j|\hat U_{F_j}^\dagger\hat U_{t_j,t_{j-1}}\hat U_{F_{j-1}}|k_{j-1}\rangle|
    \right)^2\\
\nonumber
    &=\left(\sum_{k_n,\ldots,k_0=1}^d\sqrt{\rho(k_0)}\prod_{j=1}^n 
        |\langle k_j|\hat U_{F_j}^\dagger\hat U_{t_j,t_{j-1}}
        \Big(\hat U_{F_j}\sum_{k'=1}^d |k'\rangle\langle k'|\hat U_{F_j}^\dagger\Big)\hat U_{F_{j-1}}|k_{j-1}\rangle|
    \right)^2\\
\nonumber
    &\leq \left(\sum_{k_n,\ldots,k_0=1}^d\sum_{k_n',\ldots,k_1'=1}^d\sqrt{\rho(k_0)}\prod_{j=1}^n 
        |\langle k_j|\hat U_{F_j}^\dagger\hat U_{t_j,t_{j-1}}
        \hat U_{F_j}|k_j'\rangle|\,|\langle k_j'|\hat U_{F_j}^\dagger\hat U_{F_{j-1}}|k_{j-1}\rangle|
    \right)^2\\
\nonumber
    &\leq \left(\sum_{k_n,\ldots,k_0=1}^d\sqrt{\rho(k_0)}\prod_{j=1}^n 
        \e^{(d-1)\int_{t_{j-1}}^{t_j}\|\hat H(s)\|_\mathrm{op}\mathrm{d}s}|\langle k_j|\hat U_{F_j}^\dagger\hat U_{F_{j-1}}|k_{j-1}\rangle|
    \right)^2\\
    &\leq d^2\e^{2(d-1)\int_I \|\hat H(s)\|_\mathrm{op}\mathrm{d}s}
        \left(\sum_{k_n,\ldots,k_0=1}^d\prod_{j=1}^n|\langle k_{j}|\hat U_{F_j}^\dagger\hat U_{F_{j-1}}|k_{j-1}\rangle|\right)^2,
\end{align}
where $\hat U_{F_0} = \hat U_\rho$ such that $\hat\rho\hat U_\rho|k\rangle = \rho(k)\hat U_\rho|k\rangle$. If there were only one observable (or if all observables commuted with other), then $\hat U_{F_j}^\dagger\hat U_{F_{j-1}} = \hat 1$ for all $j$ and one would recover the single-observable uniform bound~\eqref{eq:uniform_bound}. In general, $|\langle k_j|\hat U_{F_j}^\dagger\hat U_{F_{j-1}}|k_{j-1}\rangle|$ is greater than $\delta_{k_j,k_{j-1}}$, and in extreme cases it may even reach a $k_j,k_{j-1}$-independent value $1/\sqrt{d}$, which then leads to an additional factor $d^2$. Consequently, the above estimate depends explicitly on the number of arguments $n$ of the distribution, and it is thus unfit for obtaining a \textit{uniform} bound for the family.

Of course, if one cannot demonstrate that $\mathbb{Q}_{I|\rho}^{\{F\}}$ is uniformly bounded, then one cannot apply here the extension theorem to find the corresponding master bi-measure. Nevertheless, there is still a way to show that every multi-observable bi-probability $Q_{\bm{t}_n}^{\bm{F}_n}\in\mathbb{Q}_{I|\rho}^{\{F\}}$ is a restriction of a certain form of master bi-measure (compare with Section~\ref{sec:the_matter_of_time}). To this end we introduce two new notions. First are the \textit{generic bi-probabilities},
\begin{align}
    Q_{\bm{\hat{U}}_n}(\bm{k}^+_n,\bm{k}^-_n) &:= \delta_{k_n^+,k_n^-}\delta_{k_1^+,k_1^-}\operatorname{tr}\Big[
        \Big(\prod_{j=n}^1 \hat U_j|k_j^+\rangle\langle k_j^+|\hat U_j^\dagger\Big)
        \Big(\prod_{j=1}^n \hat U_j|k_j^-\rangle\langle k_j^-|\hat U_j^\dagger\Big)
    \Big].
\end{align}
These are indeed bi-probabilities because the family of such functions, together with the family of their diagonal parts, fulfills the properties~\ref{prop:properties_quantum_bi-probs} and~\ref{prop:diagonal_part}. Also, this is the most generic form of quantum bi-probability, as any multi-observable bi-probability (including the single-observable variants) can be expressed as a linear combination of $Q_{\bm{\hat{U}}_n}$,
\begin{align}
\nonumber
    Q_{\bm{t}_n}^{\bm{F}_n}(\bm{f}_n^+,\bm{f}_n^-|\rho) &=\sum_{\bm{k}_{n+1}^\pm}\prod_{j=1}^n
        \delta\big(\langle k_{j+1}^+|\hat U_{F_j}^\dagger\hat F_{j}\hat U_{F_j}|k_{j+1}^+\rangle - f_j^+\big)
        \delta\big(\langle k_{j+1}^-|\hat U_{F_j}^\dagger\hat F_{j}\hat U_{F_j}|k_{j+1}^-\rangle - f_j^-\big)\\
\label{eq:mulit-obs_bi-prob_decomp}
    &\phantom{=}\times
        \sqrt{\rho(k_1^+)\rho(k_1^-)}\,Q_{\hat U_{0,t_n}\hat U_{F_n},\,\ldots\,,\hat U_{0,t_1}\hat U_{F_1},\,\hat U_\rho}(\bm{k}_{n+1}^+,\bm{k}_{n+1}^-),
\end{align}
where $\hat\rho\hat U_\rho|k\rangle = \rho(k)\hat U_\rho|k\rangle$. The second is a \textit{path} in the space of unitary transformations (i.e., the Lie group $\operatorname{SU}(d)$),
\begin{align}
    \hat\gamma :\  [0,1]\in\tau \mapsto \mathcal{T}\e^{-\i\int_0^\tau \hat V(s)\mathrm{d}s}\hat U_0,
\end{align}
where $\hat V(s)$ is a Hermitian operator-valued function on $[0,1]$, and $\hat U_0$ is a fixed unitary operator. We combine the two into the family of generic bi-probabilities associated with a path,
\begin{align}
    \mathbb{Q}_{\gamma} &:= \big\{ Q_{\hat{\bm{U}}_n}(\bm{k}^+,\bm{k}^-)\ :\ 
        n\in\mathbb{N},
        \,\forall_{j=1,\ldots,n}\ \tau_j\in [0,1],
        \,\tau_1<\cdots<\tau_n,
        \,\hat U_j = \hat\gamma(\tau_j)
    \big\}.
\end{align}
The family is clearly bi-consistent, and it is also uniformly bounded,
\begin{align}
    \sup\Big\{\big\|Q_{\hat{\bm{U}}_n}\big\|_1\ :\ Q_{\hat{\bm{U}}_n}\in\mathbb{Q}_\gamma\Big\}
    \leqslant d^2\e^{2(d-1)\operatorname{Length}(\gamma)},
\end{align}
where the length of path $\gamma$ is given by
\begin{align}
    \operatorname{Length}(\gamma) &:= \int_0^1\bigg\|\frac{\mathrm{d}}{\mathrm{d}\tau}\hat\gamma(\tau)\bigg\|_\mathrm{op}\mathrm{d}\tau = \int_0^1\|\hat V(\tau)\|_{\mathrm{op}}\,\mathrm{d}\tau.
\end{align}
This bound is immediately obtained by noticing that $\mathbb{Q}_\gamma$ equals $\mathbb{Q}_{[0,T]}^F$ provided that $\hat V(s) = T\hat H(T s)$, $\hat U_0 = \hat U_F$, and $\hat\rho = \hat U_F|k_0\rangle\langle k_0|\hat U_F^\dagger$. Therefore, any family of generic bi-probabilities associated with a path of finite length extends to a master bi-measure $\mathcal{Q}_\gamma$ on the space of \textit{field} pairs,
\begin{align}
    \operatorname{SU}(d)\supset\operatorname{Img}(\gamma)\ni \hat U \mapsto \big( k^+(\hat U) , k^-(\hat U)\big)\in \{1,\ldots,d\}\times\{1,\ldots,d\}.
\end{align}
Conversely, when a given generic bi-probability $Q_{\bm{\hat{U}}_n}$ is found to be a member of a family associated with a certain path $\gamma$---i.e., there exist $\tau_1<\ldots<\tau_n$ such that $\hat U_j = \hat\gamma(\tau_j)$ for all $j=1,\ldots,n$---then this bi-probability is a restriction of the master bi-measure $\mathcal{Q}_\gamma$ that extends $\mathbb{Q}_\gamma$,
\begin{align}
    Q_{\hat{\bm{U}}_n}(\bm{k}^+_n,\bm{k}^-_n) &= \iint \Big(
        \prod_{j=1}^n
        \delta_{k^+(\hat\gamma(\tau_j)),k^+_j}\,\delta_{k^-(\hat\gamma(\tau_j)),k^-_j}
        \Big)\mathcal{Q}_\gamma[k^+,k^-][\mathcal{D}k^+][\mathcal{D}k^-].
\end{align}
Using this construction we can conclude that each multi-observable bi-probability is a restriction of a master bi-measure. A given bi-probability $Q_{\bm{t}_n}^{\bm{F}_n}$ decomposes into the linear combination of generic bi-probabilities, as demonstrated in Eq.~\eqref{eq:mulit-obs_bi-prob_decomp}. For every constituent bi-probability $Q_{\bm{\hat{U}}_n}$, there exists a path $\gamma_{\bm{\hat{U}}_n}$ such that $Q_{\bm{\hat{U}}_n}\in\mathbb{Q}_{\gamma_{\bm{\hat{U}}_n}}$. Then, the paths $\gamma_{\bm{\hat{U}}_n}$ can be ``stitched'' together into a master path $\Gamma$ (or, equivalently, they can be viewed as segments of $\Gamma$), so that each $\mathbb{Q}_{\gamma_{\bm{\hat{U}}_n}}\subset\mathbb{Q}_{\Gamma}$. And so, each bi-probability $Q_{\bm{\hat{U}}_n}$ becomes a restriction of the corresponding master bi-measures $\mathcal{Q}_{\gamma_{\bm{\hat{U}}_n}}$, that in turn are all restrictions of the one master bi-measure $\mathcal{Q}_{\Gamma}$ which extends $\mathbb{Q}_{\Gamma}$. Since each constituent generic bi-probability is a restriction of $\mathcal{Q}_\Gamma$, it follows that so is $Q_{\bm{t}_n}^{\bm{F}_n}$. This reasoning is a direct generalization of the result obtained in Sec.~\ref{sec:the_matter_of_time}.

The above line of argumentation is however incomplete; there are still several yet unresolved issues and open questions that, hopefully, shall be addressed in the future. For example, one such question is whether it is possible to construct a single \textit{grandmaster} path $\Sigma$ such that all $Q_{\bm{t}_n}^{\bm{F}_n}\in\mathbb{Q}_{I|\rho}^{\{F\}}$ are restrictions of the same $\mathcal{Q}_\Sigma$. Attempting to find an answer, one could opt for a ``brute force'' solution where $\Sigma$ is simply obtained by stitching together all master paths $\Gamma$ considered previously. It is however unclear whether the corresponding $\mathbb{Q}_\Sigma$ would actually be bounded, and thus, be extendable to $\mathcal{Q}_\Sigma$---it all depends on whether the length of such a $\Sigma$ is finite. Moreover, even if this solution is correct, then one still has to grapple with the problem of uniqueness of the grandmaster bi-measure $\mathcal{Q}_\Sigma$. Indeed, note that in the path-stitching method we have employed here so liberally, the order at which the constituent paths are being connected is arbitrary. Therefore, one can stitch together a possibly infinite number of grandmaster paths that only differ by the ordering of master paths $\Gamma$. What are the relations between grandmaster bi-measures corresponding to those different orderings is yet another open question.

\subsection{Connection to the quantum comb formalism}\label{subsec:quantum_combs}
%overview
As mentioned in the introduction, there are several studies concerning generalized counterparts of stochastic processes in quantum theory~\cite{Milz2020,Chiribella2009,Lindblad1979,Accardi1982}, together with generalizations of the Kolmogorov extension theorem adapted to such different descriptions of quantum processes. For instance, Accardi et al. prove a generalized Kolmogorov theorem for quantum stochastic process in the language of operator algebras, where a probability (or a correlation function) is defined as a linear functional over sequential products of quantum instruments. Similarly, a Kolmogorov extension theorem for the so-called \textit{quantum combs}\footnote{
Here we focus our discussion on the comparison between the quantum comb and the bi-probability formalisms. The connections between the bi-probability formalism and other trajectory-based formulations can be conceptually conducted in the similar fashion.
} or \textit{quantum process matrices}, was recently shown, with a more general index set, in~\cite{Milz2020}. Those are probability functionals on a sequence, or a mixture of sequences, of completely positive and trace non-increasing maps on some state spaces. 

In principle, there are similarities and differences between the bi-probability and comb formulations, in both fundamental and technical aspects, which we shall elucidate in the following. We stress, however, that the extension theorem for combs presented in~\cite{Milz2020} and the extension theorem presented in the current paper are different. Essentially, the extension theorem for combs states that the consistent family of $n$-step combs (see below) can be derived from a {\em master comb}. Equivalently, it states that there exists a quantum stochastic process~\cite{Lindblad1979,Accardi1982} which reproduces the family of $n$-step combs (equivalently, correlation kernels~\cite{Accardi1982}). Here, the quantum process is already defined in terms of evolution operator $\hat U_{t,0}$ generated by the Hamiltonian $\hat H(t)$, that is, given $\hat U_{t,0}$ (or $\hat H(t)$) one defines a family of bi-probabilities $\mathbb{Q}_I^F$. We proved that this family can be reconstructed from a master, complex-valued bi-probability measure $\mathcal{Q}[f^+,f^-][\mathcal{D}f^+][\mathcal{D}f^-]$ on the space of trajectory  $t\mapsto (f^+(t),f^-(t))$.

%construction of quantum combs
Let us briefly discuss the construction of quantum combs and the corresponding extension theorem. Let $\mathcal{H}$ be a finite-dimensional Hilbert space, and $\mathcal{B}(\mathcal{H})$ the Banach space of bounded linear operators on $\mathcal{H}.$ A completely positive (CP) and trace non-increasing map is a map $\mathcal{M}(m):\mathcal{B}(\mathcal{H})\rightarrow \mathcal{B}(\mathcal{H})$  such that $\mathcal{M}(m)\otimes\mathcal{I}$ is a positive map, where $\mathcal{I}$ is the identity on a Hilbert space $\mathbb{C}^d$ of arbitrary dimension $d.$ A sequence of CP maps can be represented by an operator $\mathcal{M}_{t_n,\ldots,t_1}(m_n,\ldots,m_1)=\bigotimes_{j=n}^1\mathcal{M}_{t_j}(m_j)$ on a product space $\mathcal{B}(\mathcal{H})^n$, where the subscript $t_j$ denotes the time at which the map $\mathcal{M}_{t_j}(m_j)$ acts. For a given time step $j$, if $\mathcal{M}_{t_j}:=\sum_{m_j}\mathcal{M}_{t_j}(m_j)$ is trace-preserving, we will call it an instrument.

%define quantum comb
In this setting, a quantum comb is a positive linear functional $\mathcal{T}$ on a set of all CP maps $\mathcal{M}_{t_n,\ldots,t_1}(m_n,\ldots,m_1)$ taking values in $[0,1]$ and such that $\mathcal{T}[\mathcal{M}_{t_n,\ldots,t_1}(m_n,\ldots,m_1)]=1$ if and only if all $\mathcal{M}_{t_j}(m_j)$ are identity maps. For instance, for a sequential measurement, one can assign $\mathcal{M}_{t_j}(m_j)$ as a measurement apparatus acting at a time step $j$ and associating with an outcome $m_j.$ In this sense, a quantum comb can describe the joint probability of a sequence of outcomes $(m_n,\ldots,m_1),$ given a sequence of measurement apparatuses $(\mathcal{M}_{t_n}(m_n),\ldots,\mathcal{M}_{t_1}(m_1))$: namely, 
\begin{equation}
    P_{t_n,\ldots,t_1}(m_n,\ldots,m_1)=\mathcal{T}\Big[\mathcal{M}_{t_n,\ldots,t_1}(m_n,\ldots,m_1)\Big].
\end{equation}
A similar idea can be extended to a probabilistic measurement by considering a mixture of different operators $\mathcal{M}_{t_n,\ldots,t_1}(m_n,\ldots,m_1).$  

%technical differences
Since we assume finite dimensions, by Riesz' representation theorem one can find a trace-$1$ positive semidefinite matrix $\mathbf{Q}$ in $M_{d^{2n}}(\mathbb{C})$ such that 
\begin{equation}\label{eq:comb-sequence}
    \mathcal{T}\Big[\mathcal{M}_{t_n,\ldots,t_1}(m_n,\ldots,m_1)\Big]= (\mathbf{Q},\mathbf{M}_{t_n,\ldots,t_1}(m_n,\ldots,m_1)),
\end{equation}
 where $\mathbf{M}_{t_n,\ldots,t_1}(m_n,\ldots,m_1)$ is a matrix representation of $\mathcal{M}_{t_n,\ldots,t_1}\!(m_n,\ldots,m_1)$ in $M_{d^{2n}}(\mathbb{C})$ and $(\cdot,\cdot)$ is an inner product on that matrix space. Recalling the sequential measurement picture outlined above, one can now define an observable \begin{align}\mathbf{X}=\sum_{m_n,\ldots,m_1}x(m_n,\ldots,m_1)\mathbf{M}_{t_n,\ldots,t_1}(m_n,\ldots,m_1),\end{align} where $x(m_n,\ldots,m_1)$ is a complex-valued function such that the matrix $\mathbf{X}$ is Hermitian. In other words, $\mathbf{M}_{t_n,\ldots,t_1}(m_n,\ldots,m_1)$ represents a measurement of the observable $\mathbf{X}$, and the average of $\mathbf{X}$ can written as
    \begin{equation}
        \operatorname{E}_{\mathbf{Q}}\big[\mathbf{X}\big] = (\mathbf{Q},\mathbf{X}) \label{eq:average-comb}.
    \end{equation}
This is similar to the average functional Eq.~\eqref{eq:average} used in the sketch of the proof for Theorem~\ref{thm:extension_theorem}. 

%average function exemplifies the connection to the comb
The connection between the formulation of quantum combs and the bi-probability one is made manifest by the average function~\eqref{eq:average-comb}: one can consider Eq.~\eqref{eq:average} in a different way by defining a matrix $\mathbf{Q}$ through $[\mathbf{Q}]_{\bm{f}_n^-,\bm{f}_n^+}:=Q_{\bm{t}_n}(\bm{f}_n^+,\bm{f}_n^-)^*$ using the trajectory set as an index set for the canonical coordinate, and also defining $\mathbf{X}$ through $[\mathbf{X}]_{\bm{f}_n^+,\bm{f}_n^-}:=X_{\bm{t}_n}(\bm{f}_n^+,\bm{f}_n^-).$ The average~\eqref{eq:average} can be represented as an inner product of these two matrices, i.e.~$\operatorname{E}_{\bm{t}_n}[X_{\bm{t}_n}]=\operatorname{tr}(\mathbf{Q}^\dagger\cdot\mathbf{X}),$ in a similar fashion as in Eq.~\eqref{eq:average-comb} in the quantum comb formulation. 

%recall the notion of elementary trajectory
On the other hand, it is worth noting that the quantum comb formalism outlined above and the one developed in the present paper exhibit a crucial difference in the way they pursue a generalization of the trajectory picture in quantum theory. Let us elaborate on that point. In classical stochastic processes, one only has single trajectories which, in a standard quantum description, can be represented by a sequence of projection operators, possibly alternated by evolution maps. These \textit{elementary} trajectories and their mixtures can describe all possible classical stochastic processes; however, quantum mechanics allows processes that cannot be represented in this way~\cite{Sakuldee2018}. Formally, such a process is represented by a ``non-classical'' trajectory in the form of a sequence of general CP maps that \textit{cannot} be written as a mixture of elementary trajectories. This is the type of generalization of the trajectory concept customarily employed in the literature, e.g. in~\cite{Lindblad1979,Accardi1982,Chiribella2009}, including the comb formulation in~\cite{Milz2020}. The present work explores an \textit{alternative} direction---to keep the notion of elementary trajectories by keeping the intervention of the projective type, and ``doubling'' the trajectories to a bi-trajectory picture. Within this picture, one still preserves the fundamental aspect of classical mechanics but one needs two trajectories to describe quantum phenomena. Let us remark that in the first generalization, one can obtain two (non-classical) trajectories by a polarization technique, but they are not necessarily elementary in the classical sense, i.e., each trajectory in the pair does not need to correspond to a sequence of projective measurements.

%similarities in the holistic sense
Although the origins of both formulations are motivated by different generalizations of trajectory pictures in quantum mechanics, on the collective level they represent the same physical scenarios. This can be seen from the definitions of the average functions, cf.~Eq.~\eqref{eq:average} and Eq.~\eqref{eq:average-comb}, which show no notions of representation dependence. Indeed, one can extend the quantum comb formulations to cover the bi-probabilities and vice versa. For instance, from the non-classical trajectory picture, one relaxes the CP condition and allows the interventions to be non-CP. In this sense, one can employ a \textit{bi-instrument}
    \begin{equation}\label{eq:N-maps}
        \mathcal{N}_{t_j}(f^+_j,f^-_j)\hat A := \hat P_{t_j}(f^+_j) \hat A \hat P_{t_j}(f^-_j),
    \end{equation}
which is non-CP but satisfies the consistency property $\sum_{f^+_j,f^-_j} \mathcal{N}_{t_j}(f^+_j,f^-_j)=\mathcal{I}$, where $\mathcal{I}$ is an identity channel. Using this type of maps in a similar construction to that of quantum combs, it follows that
    \begin{equation}\label{eq:T-to-Q}
        \mathcal{T}[\mathcal{N}_{t_n}(f^+_n,f^-_n)\otimes\cdots\otimes\mathcal{N}_{t_1}(f^+_1,f^-_1)] = Q_{t_n,\ldots,t_1}(f_n^+,f_n^-;\ldots;f_1^+,f_1^-). 
    \end{equation}
Note that, by allowing non-CP maps, one will lose the notion of trajectories, resolving in the appearance of bi-trajectories.

%consistency
Apart from the arguments above, the corresponding generalizations of the Kolmogorov extension theorem show clear similarities. One can observe that bi-consistency for quantum bi-probabilities is the same as the consistency condition for quantum combs in Ref.~\cite{Milz2020}. This simply follows from the fact that the sum of projective operators both on the right and the left-hand sides is an identity channel, as explicitly written in Eq.~\eqref{eq:N-maps}. 
As for the uniform boundedness, it can be seen that the constraint $(\mathbf{Q},\mathbf{M}_{t_n,\ldots,t_1}(m_n,\ldots,m_1))\leq 1$ for the generalized extension theorem (GET) in Ref.~\cite{Milz2020} can also imply our assumption~\ref{assump:extendable_family:uni_bound}. In particular, we know from the H{\"o}lder inequality that
    \begin{equation}
        0\leq(\mathbf{Q},\mathbf{M}_{t_n,\ldots,t_1}(m_n,\ldots,m_1))\leq\|\mathbf{Q}\|_1 \|\mathbf{M}_{t_n,\ldots,t_1}(m_n,\ldots,m_1)\|_\infty.
    \end{equation}
Up to some scaling, one may expect the connection between the uniform boundedness and the constraint $(\mathbf{Q},\mathbf{M}_{t_n,\ldots,t_1}(m_n,\ldots,m_1))\leq 1$ used for GET, and hence the connection between our results on the extension theorem for the bi-trajectory picture and GET for the non-classical trajectory picture would be established. We leave a detailed analysis of this aspect for further studies.

\section{Concluding remarks}\label{sec:conclusion}

We have shown that the multitime complex bi-probability distributions $Q_{\bm{t}_n}(\bm{f}_n^+,\bm{f}_n^-)$---and thus the corresponding (positive-valued) probability distributions $P_{\bm{t}_n}(\bm{f}_n)$---associated with an observable $F$ of a finite-dimensional quantum system are a discrete-time restrictions of a unique complex \textit{master} bi-measure $\mathcal{Q}[f^+,f^-][\mathcal{D}f^+][\mathcal{D}f^-]$ on the space of \textit{pairs} of trajectories traced through the spectrum of $F$. Mathematically, this follows from a conceptually simple, yet nontrivial generalization of Kolmogorov's extension theorem to the complex-valued scenario. Consequently, the multitime probabilities $P_{t_n,\ldots,t_1}(f_n,\ldots,f_1)$ used to describe sequential projective measurements of quantum observable do not correspond to a sampling of an underlying trajectory (akin to classical theories), but rather they result from a superposition of trajectory \textit{pairs}, $t\mapsto(f^+(t),f^-(t))$, such that $f^+(t_j)=f^-(t_j)=f_j$. Therefore, in the bi-trajectory picture, the classical limit is achieved when the trajectory pairs that do not overlap \textit{completely} have vanishing measure.

Our results may pave the way to diverse research directions. The most natural one would be the extension of our results to the infinite-dimensional scenario, in which the spectrum of the observable $F$ may be unbounded and/or admit continuous components; while the generalization of Kolmogorov's extension theorem does not forbid, in principle, either possibility (see Theorem~\ref{thm:extension_theorem_pro} in the appendix), a nontrivial refinement of the bounds presented in Section~\ref{sec:quantum_extension} is needed to this purpose.

Besides, it is our firm belief that the proof of the existence of a master bi-measure---and thus, of the bi-trajectory picture---must entail some physical ramifications for the quantum theory. Admittedly, the analysis of these potential physical implications is outside the scope of the present work, but it is a natural direction for future investigations. However, it is likely that the comprehensive assessment of the impact of our result is pending further research; to see why, let us review the ``program'' we were following in this work: The starting point was the phenomenology of quantum mechanics that gives us families of inconsistent multitime probability distributions. Taking the inspiration from classical theories, we hypothesized that the individual probabilities are connected by being components of some \textit{master} object that describe to some extent the physical system in question. The hypothesis was then confirmed with the constructive mathematical proof of the object's existence. This last step, the mathematical inference, is the only real way to uncover the master object's form---the object itself cannot be directly ``seen'', only its components, the probabilities, are accessible. Here, we were able to carry out this program to its conclusion---and thus, we have found the master bi-measure---but only in the case of single observable. As indicated in Section~\ref{subsec:multiple_observables}, we were unable to do the same for the multi-observable scenario. Consequently, we cannot be sure what is the exact nature of the \textit{true} master object, the single object that is the source of all phenomenologically accessible multi-observable probabilities (including the single observable as a special case). We believe that this true master object does exist, but the question of its form remains; only a well-constructed mathematical proof will be able to provide the answer.

\section*{Acknowledgments}

D.L. acknowledges financial support from the PNRR MUR project CN00000013 - National Centre for HPC, Big Data and Quantum Computing, and partial support by Istituto Nazionale di Fisica Nucleare (INFN) through the project “QUANTUM” and the Italian National Group of Mathematical Physics (GNFM-INdAM). 
F.S. acknowledges support by the Foundation for Polish Science (IRAP project, ICTQT, contract no. 2018/MAB/5, co-financed by EU within Smart Growth Operational Programme). 
D.C. was supported by the Polish National Science Center project No. 2018/30/A/ST2/00837. 
P.S. acknowledges support from the Polish National Science Centre through the project MAQS under QuantERA, which has received funding from the European Union’s Horizon 2020 research and innovation program under Grant Agreement No. 731473, Project No. 2019/32/Z/ST2/00016. P.S. also extends sincere gratitude to J. Korbicz and J. Wehr for their invaluable advice and insightful discussions; they have been immensely helpful during the development of this work.

\appendix

\section{Kolmogorov extension theorem for complex measures}\label{app:measure_theory}

In this appendix we will collect some basic notions about complex measures, and we will ultimately state and prove an extension theorem (Theorem~\ref{thm:extension_theorem_pro}) for consistent families of complex measures, which also reduces to the ``standard'' Kolmogorov extension theorem in the case in which all measures are positive-valued. Theorem~\ref{thm:extension_theorem} in the main text will henceforth follow as a particular case of Theorem~\ref{thm:extension_theorem_pro}. Further details on measure theory can be found in any of the standard monographs on the subjects---among them, we mainly refer to Rudin's classic textbook~\cite{rudin1986real} for the main concepts, and Bhattacharya and Waymire's book~\cite{bhattacharya2017basic} for Kolmogorov's extension theorem and related concepts.  

The aim of this discussion (which, needless to say, has no pretense of completeness) is twofold. From one hand, this appendix could be regarded as a quick ``handbook'' of measure-theoretical notions of immediate use for the purpose of the paper---incidentally, also providing a relatively simple, functional analysis--oriented proof of the Kolmogorov extension theorem which should be of interest for quantum physicists even beyond the scope of this paper. From the other hand, to prove the desired extension theorem, we will need to explicitly state some properties of complex measures which, in the literature, are often only discussed in the positive-valued scenario. For the convenience of the reader, we therefore decided to provide an explicit, if concise, discussion of complex measures.

\subsection{Basic notions on complex measures}

\begin{definition}[Measurable space]\label{def:sigma_algebra}
			A \textit{measurable space} is a pair $(M,\mathcal{F})$, where $M$ is a set and $\mathcal{F}$ is a \textit{$\sigma$-algebra} on $M$, that is, a family of subsets of $M$ such that
			\begin{itemize}
				\item $M\in\mathcal{F}$;
				\item $A\in\mathcal{F}$ iff $M\setminus A\in\mathcal{F}$;
				\item given a countable family $\{A_i\}_{i\in\mathbb{N}}\subset\mathcal{F}$, we have $\bigcup_{i\in\mathbb{N}}A_i\in\mathcal{F}$.
			\end{itemize}
	$M$ is called \textit{sample space}, and the elements of $\mathcal{F}$ are its \textit{measurable} subsets.
	\end{definition}
	It is easy to see that $\mathcal{F}$ is also closed with respect to countable intersections, and that the null set $\emptyset$ also belongs to $\mathcal{F}$. By construction, the intersection of $\sigma$-algebras is also a $\sigma$-algebra itself. An obvious example of $\sigma$-algebra on $M$ is the class of \textit{all} its subsets---its \textit{power set}: this is indeed the standard choice when dealing with spaces of finite cardinality. Nevertheless, to define measures on spaces of trajectories, we shall also need to work with spaces with infinite cardinality, in which case, for technical reasons, other choices of $\sigma$-algebras are to be made.
   
   Given a family $F$ of subsets of $M$, there always exists a $\sigma$-algebra containing all elements of $F$ (e.g.~the power set itself). Then the intersection of all $\sigma$-algebras containing $F$ is still a $\sigma$-algebra containing $F$---by construction, the \textit{smallest} one. We will call it the $\sigma$-algebra \textit{generated} by $F$. A particularly useful example is the following: whenever $M$ is endowed with a \textit{topological space} structure (e.g., $\mathbb{C}$ or $\mathbb{R}$), the \textit{Borel $\sigma$-algebra} of $M$, which we shall denote as $\mathfrak{B}(M)$, is the $\sigma$-algebra generated by the open subsets of $M$. In such a case, all elements of $\mathfrak{B}(M)$ are denoted as the \textit{Borel (sub)sets} of $M$.

    The notion of $\sigma$-algebra generated by a family of subsets is also involved in the following definition, of primary importance for our purposes:
   \begin{definition}[Product space and product $\sigma$-algebra]\label{def:product_space}
       Let $(M_t)_{t\in I}$ an arbitrary collection of sets, each endowed with a $\sigma$-algebra $\mathcal{F}_t$. The \textit{product space} of $(M_t)_{t\in I}$ is defined as
       \begin{equation}
           \prod_{t\in I}M_t=\left\{f=(f_t)_{t\in I}:\;f_t\in M_t\right\},
       \end{equation}
       that is, the space of functions from $I$ to $\cup_{t\in I}M_t$. Besides, the \textit{product $\sigma$-algebra} of $(\mathcal{F}_t)_{t\in I}$, denoted by $ \bigotimes_{t\in I}\mathcal{F}_t$, is the $\sigma$-algebra on $\prod_{t\in I}M_t$ generated by all \textit{finite-dimensional rectangles}, that is, all subsets of $\prod_{t\in I}M_t$ in the form
       \begin{equation}
           \left\{f\in\prod_{t\in I}M_t:\;(f_{t_1},\dots,f_{t_n})\in A_{t_1}\times\cdots\times A_{t_n},\;A_{t_1}\in\mathcal{F}_{t_1},\ldots,A_{t_n}\in\mathcal{F}_{t_n}\right\}.
       \end{equation}
  \end{definition}
With these definitions, the pair $(\prod_{t\in I}M_t,\otimes_{t\in I}\mathcal{F}_t)$ is itself a measurable space. In the definition above, $I$ (which will eventually play the role of time in our applications) can either be finite, countably finite, or having the cardinality of the continuum. 

\begin{definition}[Complex-valued measure]\label{def:measure}
			Let $(M,\mathcal{F})$ a measurable space. A \textit{complex-valued measure} (or complex measure) on $(M,\mathcal{F})$ is a function $\nu:\mathcal{F}\rightarrow\mathbb{C}$ which satisfies the following properties:
			\begin{itemize}
				\item[(i)] \textit{grade-1 additivity}: given two \textit{disjoint} sets $A_1,A_2\in\mathcal{F}$,
				\begin{equation}
					\nu(A_1\sqcup A_2)=\nu(A_1)+\nu(A_2);
				\end{equation}
				\item[(ii)] \textit{continuity}: given an increasing sequence of sets $\{A_i\}_{i\in\mathbb{N}}$ (that is, $A_{i}\subset A_{i+1}$),
				\begin{equation}\label{eq:continuity}
					\lim_{i\to\infty}\nu(A_i)=\nu\left(\bigcup_{i\in\mathbb{N}}A_i\right);
				\end{equation}
			\end{itemize}
   If $M$ is a topological space and $\mathcal{F}:=\mathfrak{B}(M)$ its corresponding Borel $\sigma$-algebra, we will refer to $\nu$ as a \textit{Borel} measure.
		\end{definition}
Note that we are using the symbol $\sqcup$ for the union of \textit{disjoint} sets: this should be understood anytime such symbols appear.	By iterating (i), one obtains, for all $n\geq2$,
		\begin{equation}
			\nu\left(\bigsqcup_{i=1}^nA_i\right)=\sum_{i=1}^n\nu(A_i),
		\end{equation}
	whence the measure of the union of any family of disjoint sets can be entirely reconstructed by knowing the measure associated with each separate set; furthermore, (ii) allows us to take the limit $n\to\infty$ in the expression above ($\sigma$-additivity).

    If $\nu(A)$ is positive for all $A\in\mathcal{F}$, one recovers the familiar definition of \textit{finite} positive-valued measures---in fact, without further specification, one often reserves the noun ``measure'' for the positive-valued case. The adjective ``finite'' is here due to the fact that, differently from the complex case, positive-valued measures may admit infinite values.   
    In the remainder of this work we shall always take for granted that all measures are finite and possibly complex-valued.

The following definition is of major importance for complex measures:
\begin{definition}\label{def:variation}
Let $(M,\mathcal{F})$ be a measurable space, and $\nu$ a complex-valued measure on it. Its \textit{variation} $|\nu|$ is the positive measure defined as follows:
\begin{equation}
    |\nu|(A)=\sup\left\{\sum_k|\nu(A_k)|:(A_k)_k\text{ pairwise disjoint, }\bigsqcup_kA_k=A\right\}.
\end{equation}
Its \textit{total variation} is the quantity $|\nu|(M)$.
\end{definition}
It can be proven that this is indeed a positive measure on $(M,\mathcal{F})$, and $|\nu(A)|\leq|\nu|(A)$~\cite[Theorem 6.2]{rudin1986real}; besides, of course $|\nu|=\nu$ whenever $\nu$ is positive. As such, the definition above only acquires importance when dealing with non-positive measures---indeed, it will be of primary importance for obtaining a Kolmogorov extension theorem beyond the positive scenario.

Proceeding in analogy with positive measures, we can likewise define a notion of \textit{integral} with respect to a complex measure. Recall that a function $f:M\rightarrow\mathbb{C}$ is said to be \textit{measurable} if, for every Borel subset $A\subset\mathbb{C}$, its preimage
\begin{equation}
    f^{-1}(A):=\{\omega\in M:\;f(\omega)\in A\}
\end{equation}
satisfies $f^{-1}(A)\in\mathcal{F}$. For such functions, we can define the \textit{Lebesgue integral} with respect to a complex measure $\nu$, denoted by
\begin{equation}
    \int_M f(x)\;\mathrm{d}\nu(x)\qquad\text{or}\qquad\int_M f\;\mathrm{d}\nu
\end{equation}
via the same procedure that is used for positive measures, that is, by approximating a measurable function with simple functions (i.e., linear combinations of indicator functions), defining the integral of such functions in the usual way, and taking the limit.

An important relation between a complex measure $\nu$ and its variation $|\nu|$ is the \textit{polar decomposition} of $\nu$, see e.g.~\cite[Theorem 6.12]{rudin1986real}: there always exists a measurable function $h:M\rightarrow\mathbb{C}$ with $|h(x)|\leq1$ such that $\mathrm{d}\nu(x)=h(x)\;\mathrm{d}|\nu|(x)$, 
which implies
\begin{equation}\label{eq:integrals}
    \int_Mf\;\mathrm{d}\nu=\int_Mfh\;\mathrm{d}|\nu|
\end{equation}
and, as a direct consequence,
\begin{equation}\label{eq:triangular}
        \left|\int_Mf\;\mathrm{d}\nu\right|\leq\int_M|f|\;\mathrm{d}|\nu|.
\end{equation}
Eq.~\eqref{eq:integrals} can be used to show that the \textit{dominated convergence theorem} also holds for complex measures. We shall state and prove this fact explicitly.
\begin{theorem}[Dominated convergence theorem]
    Let $(f_n)_{n\in\mathbb{N}}$ be a sequence of complex-valued measurable functions on $(M,\mathcal{F})$, pointwise converging to a complex-valued function $f$. Suppose that there exists a positive-valued function $g$ which is integrable with respect to a complex measure $\nu$ and such that $|f_n(x)|\leq g(x)$. Then
    \begin{equation}
        \lim_{n\to\infty}\int_M f_n\;\mathrm{d}\nu=\int_M f\;\mathrm{d}\nu.
    \end{equation}
\end{theorem}
\begin{proof}
    If $\nu$ is a positive measure, this is the usual statement of the dominated convergence theorem, see e.g.~\cite[Theorem 1.34]{rudin1986real}. In the general case, we know that $\nu$ admits a polar decomposition $\mathrm{d}\nu(x)=h(x)\;\mathrm{d}|\nu|(x)$, with $|h(x)|\leq1$; then clearly $f_nh$ converges pointwise to $fh$ and $|f_n(x)h(x)|\leq|f_n(x)|\leq g(x)$, whence, applying the dominated convergence theorem for positive measures,
    \begin{equation}
        \lim_{n\to\infty}\int_M f_nh\;\mathrm{d}|\nu|=\int_M fh\;\mathrm{d}|\nu|,
    \end{equation}
    which, by Eq.~\eqref{eq:integrals}, is the desired equality.
\end{proof}

    \subsection{Kolmogorov extension theorem for complex measures}

    We will now provide an extension theorem for consistent families of complex-valued Borel measures, akin to the ``standard'' Kolmogorov extension theorem for positive ones. The result hereby presented can be regarded as an immediate generalization of Theorem~\ref{thm:extension_theorem} (cf.~Remark~\ref{rem:pro}).
    
    Given an arbitrary set $I$, we consider a family $(M_t)_{t\in I}$ of compact metric spaces, each endowed with the corresponding Borel $\sigma$-algebra $\mathcal{F}_t:=\mathfrak{B}(M_t)$. Recalling Definition~\ref{def:product_space}, for every $n\in\mathbb{N}$ and every (not necessarily ordered) $\bm{t}_n=(t_1,\dots,t_n)\in I^n$, with $t_j\neq t_\ell$ for all $j\neq\ell$, we are automatically given a product space (cf.~Definition~\ref{def:product_space})
    \begin{equation}
        \left(M_{\bm{t}_n},\mathcal{F}_{\bm{t}_n}\right),\qquad M_{\bm{t}_n}=M_{t_1}\times\cdots\times M_{t_n},\qquad\mathcal{F}_{\bm{t}_n}=\mathcal{F}_{t_1}\otimes\cdots\otimes\mathcal{F}_{t_n},
    \end{equation}
    with each $M_{\bm{t}_n}$ being a compact metric space, and $\mathcal{F}_{\bm{t}_n}$ being the corresponding Borel $\sigma$-algebra. Besides, we can also consider the ``total'' product space
    \begin{equation}
        \left(\prod_{t\in I}M_t,\;\;\bigotimes_{t\in I}\mathcal{F}_t\right),
    \end{equation}
    whose elements, cf.~Definition~\ref{def:product_space}, can be either regarded as (possibly uncountable) families $(f_t)_{t\in I}$, with $f_t\in M_t$, or equivalently as functions $f:I\rightarrow\cup_tM_t$.
    
    \begin{definition}[Consistent family of measures]\label{def:consistency}
       Consider a family $(\nu_{\bm{t}_n})_{\bm{t}_n\in I^n}$, with $\nu_{\bm{t}_n}$ being a complex-valued Borel measure on $(M_{\bm{t}_n},\mathcal{F}_{\bm{t}_n})$. The family is said to be \textit{consistent} if it satisfies the following properties:
       \begin{itemize}
           \item[(a)] for every $n\in\mathbb{N}$, every permutation $\sigma:\{1,\dots,n\}\rightarrow\{1,\dots,n\}$, and every $A_1\in\mathcal{F}_{t_1},\ldots, A_n\in\mathcal{F}_{t_n}$, 
           \begin{equation}\label{eq:permutations}
               \nu_{t_n,\dots,t_1}(A_n\times\ldots\times A_1)=\nu_{t_{\sigma(n)},\dots,t_{\sigma(1)}}(A_{\sigma(n)}\times\ldots\times A_{\sigma(1)});
           \end{equation}
         \item[(b)] for every $n\in\mathbb{N}$, every distinct $t_1,\dots,t_n,t_{n+1}\in I$ and every $A\in\mathcal{F}_{t_1}\times\ldots\times\mathcal{F}_{t_n}$,
         \begin{equation}\label{eq:consistency_pro}
             \nu_{t_1,\ldots,t_n}(A)=\nu_{t_1,\dots,t_n,t_{n+1}}(A\times M_{t_{n+1}}).
         \end{equation}
       \end{itemize}
    \end{definition}

    \begin{remark}\label{rem:consistency}
        Eq.~\eqref{eq:consistency_pro} can be equivalently rephrased as follows: for every $n\in\mathbb{N}$, every distinct $t_1,\dots,t_n,t_{n+1}\in I$ and every $B\in\mathcal{F}_{t_1}\times\ldots\times\mathcal{F}_{t_n}$,
        \begin{equation}
            \int1_{B}\;\mathrm{d}\nu_{t_1,\dots,t_n}=\int1_{B\times M_{t_{n+1}}}\;\mathrm{d}\nu_{t_1,\dots,t_n,t_{n+1}},
        \end{equation}
        which therefore implies the equality
        \begin{equation}
          \int X_{t_1,\dots,t_n}\mathrm{d}\nu_{t_1,\dots,t_n}=\int X_{t_1,\dots,t_n}\;\mathrm{d}\nu_{t_1,\dots,t_n,t_{n+1}}
        \end{equation}
        for every simple function $X_{t_1,\dots,t_n}:M_{\bm{t}_n}\rightarrow\mathbb{C}$ and thus, via a standard approximation argument, for \textit{every measurable} function $X_{t_1,\dots,t_n}$. Besides, by Eq.~\eqref{eq:permutations}, the same holds if the role of $t_{n+1}$ is replaced by any other coordinate:  
        \begin{equation}\label{eq:consistency_pro2}
          \int X_{t_1,\ldots,\cancel{t_j},\ldots,t_{n+1}}\mathrm{d}\nu_{t_1,\ldots,\cancel{t_j},\ldots,t_{n+1}}=\int X_{t_1,\ldots,\cancel{t_j},\ldots,t_{n+1}}\mathrm{d}\nu_{t_1,\ldots,t_j,\ldots,t_{n+1}}
        \end{equation}
        and so on; that is, ``adding times'' to the integration measure does not modify the value of the integral.
    \end{remark}
    
    \begin{theorem}\label{thm:extension_theorem_pro}
        Let $(M_t)_{t\in I}$ a family of compact metric spaces, with $\mathcal{F}_t$ being the corresponding Borel $\sigma$-algebra, and $(\nu_{\bm{t}_n})_{\bm{t}_n\in I}$ a family of complex-valued Borel measures on $(M_{\bm{t}_n},\mathcal{F}_{\bm{t}_n})$ satisfying the following assumptions:
        \begin{itemize}
            \item consistency in the sense of Definition~\ref{def:consistency};
            \item uniformly bounded total variation:
            \begin{equation}\label{eq:uniform_boundedness_pro}
                \sup_{n\in\mathbb{N},\bm{t}_n\in I^n}\left|\nu_{\bm{t}_n}\right|(M_{\bm{t}_n})<\infty,
            \end{equation}
            with $|\nu_{\bm{t}_n}|$ as in Definition~\ref{def:variation}.
        \end{itemize}
        Then there exists a unique complex-valued measure $\nu$ on the product space $(\mathcal{M},\otimes_t\mathcal{F}_t)$ such that, for all $n\in\mathbb{N},\bm{t}_n\in I^n$ and $A\in\mathcal{F}_{\bm{t}_n}$,
        \begin{equation}\label{eq:cylinders_pro}
            \nu_{\bm{t}_n}(A)=\nu\left(\left\{f\in\prod_{t\in I}M_t:(f(t_1),\ldots,f(t_n))\in A\right\}\right).
        \end{equation}
    \end{theorem}
    The proof of the theorem, postponed to Section~\ref{app:proof_pro}, is an adaptation to the complex-valued case of the ``Functional Analysis Version'' of the proof of the standard Kolmogorov extension theorem~\cite[Theorem 9.1]{bhattacharya2017basic} (also see~\cite[Theorem 11.1]{Davies_76}). The assumption of compactness of the spaces $M_t$ is not strictly necessary and may be removed by following a similar argument as the one presented in both referenced proofs to cover the non-compact case.

    Eq.~\eqref{eq:cylinders_pro} states the following: given any Borel set $A\in\mathcal{F}_{\bm{t}_n}$, the value of $\nu_{\bm{t}_n}$ on it corresponds to the value of $\nu$ on the corresponding \textit{cylinder}---the set of all functions $f$ whose values at the times $t_1,\dots,t_n$ are in $A$. As such, $\nu$ unambiguously generates the full consistent family of measures $\nu_{\bm{t}_n}$. 
    Theorem~\ref{thm:extension_theorem_pro} reduces to the ``standard'' Kolmogorov extension theorem (at least in the case of compact metric spaces) when all measures involved are positive-valued---importantly, in such a case, the request~\eqref{eq:uniform_boundedness_pro} becomes trivial for probability measures since $|\nu_{\bm{t}_n}|=\nu_{\bm{t}_n}$ and all quantities in Eq.~\eqref{eq:uniform_boundedness_pro} are simply equal to one by definition. In this sense, the request of uniformly bounded total variation is \textit{the} additional ingredient required in order to extend the result to complex-valued measures.

    \begin{remark}\label{rem:pro}
        We can easily check that Theorem~\ref{thm:extension_theorem} in the main text is a corollary of Theorem~\ref{thm:extension_theorem_pro} corresponding to the case in which $M_t=\Omega\times\Omega$ (and thus $\prod_tM_t=(\Omega\times\Omega)^I$) for all $t\in I$, where $\Omega$ is a finite set of real numbers (and thus a compact metric space), in which case each $\nu_{\bm{t}_n}$ is thus given by
        \begin{equation}
            \nu_{\bm{t}_n}(A^+\times A^-)=\sum_{\bm{f}^\pm_n\in A^\pm}B_{\bm{t}_n}(\bm{f}^+_n,\bm{f}^-_n)
        \end{equation}
        for all couple of Borel sets $A^\pm$. This is shown by taking into account the following facts:
        \begin{itemize}
            \item the Borel $\sigma$-algebras on $\Omega\times\Omega$ and all products of finitely many copies of it simply coincide with the respective power sets;
            \item the integral with respect to each measure $\nu_{\bm{t}_n}$ reduces to a finite sum:
            \begin{equation}\label{eq:discrete_integral}
                \int X_{\bm{t}_n}(\bm{f}_n^+,\bm{f}_n^-)\;\mathrm{d}\nu_{\bm{t}_n}(\bm{f}_n^+,\bm{f}_n^-)=\sum_{\bm{f}_n^+,\bm{f}_n^-}X_{\bm{t}_n}(\bm{f}_n^+,\bm{f}_n^-)Q_{\bm{t}_n}(\bm{f}_n^+,\bm{f}_n^-);
            \end{equation}
            \item having Eq.~\eqref{eq:discrete_integral} into account, the consistency condition~\eqref{eq:consistency_pro2} (up to permutations) reduces to Eq.~\eqref{eq:kolmogorov_biconsistency};
            \item similarly, the uniform boundedness condition~\eqref{eq:uniform_boundedness_pro} reduces to Eq.~\eqref{eq:uniform_bound}.
        \end{itemize}
       Finally, setting $A=\{\bm{f}_n^+\times\bm{f}_n^-\}$ in Eq.~\eqref{eq:cylinders_pro}, one obtains
       \begin{align}\label{eq:cylinders_bis}
            Q_{\bm{t}_n}(\bm{f}_n^+,\bm{f}_n^-)&=\nu\Bigl(\bigl\{(f^+,f^-)\in(\Omega\times\Omega)^I:f^\pm(t_j)=f^\pm_j,\;j=1,\dots,n\bigr\}\Bigr),
       \end{align}
       and the latter quantity can be formally expressed as the integral over a product of Dirac distributions as in Eq.~\eqref{eq:gen_extension_thrm:extension_formula}, where $\mathcal{Q}_I[f^+,f^-]\,[\mathcal{D}f^+][\mathcal{D}f^-]$ stands for the measure element $\mathrm{d}\nu[f^+,f^-]$.

       Some final comments are in order. The quantity $\nu_{\bm{t}_n}$ can be equivalently interpreted, with an obvious abuse of notation, as a \textit{bi-measure} $(A^+,A^-)\mapsto\nu_{\bm{t}_n}(A^+,A^-)$, that is, a function of \textit{pairs} of sets that is additive in both arguments. If, in particular, we assume that the functions $B_{\bm{t}_n}$ satisfy the property~\ref{prop:positivity} (as is always the case for those associated with a quantum mechanical system, cf.~Section~\ref{sec:quantum_extension}), then for any finite collection $(A_j)_j$, the matrix $\{\nu_{\bm{t}_n}(A_j,A_k)\}_{j,k}$ is positive semidefinite. Bi-measures satisfying this property are sometimes called \textit{decoherence functionals} (or functions). In particular, their ``diagonal'' part of such measures, $\mu_{\bm{t}_n}(A):=\nu_{\bm{t}_n}(A,A)$, are positive-valued functions which, however, are \textit{not} additive (and thus, not measures). Instead, they satisfy a weaker condition, known as grade-2 additivity:
       \begin{align}\label{eq:grade-2}
	       \mu_{\bm{t}_n}\left(A_1\sqcup A_2\sqcup A_3\right)&=\mu_{\bm{t}_n}(A_1\sqcup A_2)+\mu_{\bm{t}_n}(A_2\sqcup A_3)+\mu_{\bm{t}_n}(A_1\sqcup A_3)\nonumber\\
        &\phantom{=}
            -\mu_{\bm{t}_n}(A_1)-\mu_{\bm{t}_n}(A_2)-\mu_{\bm{t}_n}(A_3),
        \end{align}
     and are thus known as grade-2 or \textit{quantum} measures, cf.~\cite{Sorkin_ModPhysLettA94,gudder2009quantum,gudder2010finite,Gudder_MathSlov12}.
\end{remark}

\subsection{Proof of Theorem~\ref{thm:extension_theorem_pro}}\label{app:proof_pro}

In order to prove Theorem~\ref{thm:extension_theorem_pro}, we will need to state some intermediate results that hold for \textit{regular} complex measures.    
To begin with, a complex Borel measure is said to be \textit{regular} if, for every Borel set $A$,
    \begin{equation}
        |\nu|(A)=\sup\{|\nu|(A'): A'\subset A\text{ is closed}\}=\inf\{|\nu|(A'): A'\supset A\text{ is open}\};
    \end{equation}    
    in particular, all Borel measures on a compact metric space are regular.

    If $M$ is a compact topological space and is also Hausdorff (that is, every couple of points admit disjoint neighborhoods), the space $\mathrm{C}(M)$ of complex-valued continuous functions on $M$, endowed with the supremum norm, is a Banach space. The following result, often referred to as the Riesz--Markov representation theorem, establishes an important connection between bounded functionals on $\mathrm{C}(M)$ and complex-valued measures.

    \begin{theorem}[Riesz--Markov theorem]\label{thm:riesz-markov}
        Let $M$ be a compact Hausdorff space. Then, for every bounded linear functional $\mathcal{E}:\mathrm{C}(M)\rightarrow\mathbb{C}$, there exists a unique regular complex-valued Borel measure $\nu$ such that
        \begin{equation}\label{eq:riesz_markov}
            \mathcal{E}(f)=\int_M f\;\mathrm{d}\nu;
        \end{equation}
        besides, $\mathcal{E}$ is a positive functional iff $\nu$ is a positive-valued measure. The operator norm of $\mathcal{E}$ equals the \textit{total variation} of $\nu$:
        \begin{equation}
            \left\|\mathcal{E}\right\|:=\sup_{\|f\|_\infty=1}|\mathcal{E}(f)|=|\nu|(M).
        \end{equation}
    \end{theorem}
    This is a standard result in measure theory, see e.g.~\cite[Theorem 6.19]{rudin1986real} (in a slightly more general setting). Of course, the converse of the theorem is also true---every functional in the form~\eqref{eq:riesz_markov}, for a given complex measure $\nu$, has the desired properties. This establishes a one-to-one correspondence between bounded linear functionals on $\mathrm{C}(M)$ and regular complex Borel measures on $M$.

    We will also need the following property which, while elementary, is usually stated for positive measures only---as such, we shall provide an explicit proof for it in the general case:
    \begin{lemma}\label{lemma:measure_determining}
        Let $M$ be a metric space, and $\nu_1$, $\nu_2$ two complex regular Borel measures such that the equality
        \begin{equation}
            \int_M f\;\mathrm{d}\nu_1=\int_Mf\;\mathrm{d}\nu_2 
        \end{equation}
        holds for all real-valued continuous functions $f$ on it. Then $\nu_1=\nu_2$.
    \end{lemma}
    The proof of this statement is analogous to the proof of~\cite[Proposition 1.6]{bhattacharya2017basic} (for positive measures); we shall sketch it below.
    \begin{proof}
         Following the construction in the proof of~\cite[Proposition 1.6]{bhattacharya2017basic}, one shows that, given any \textit{closed} Borel set $A\subset M$, then there exists a sequence of nonnegative continuous functions $(f_n)_{n\in\mathbb{N}}$ which converge pointwise to the indicator function $1_A$ of $A$. By hypothesis, for all $n\in\mathbb{N}$ we have
         \begin{equation}
             \int_M f_n\;\mathrm{d}\nu_1=\int_M f_n\;\mathrm{d}\nu_2;
         \end{equation}
         applying the dominated convergence theorem yields $\nu_1(A)=\nu_2(A)$, whence the two measures $\nu_1$ and $\nu_2$ agree on all closed Borel sets. 
         
         To obtain the desired equality for an arbitrary Borel set $A$, we note that, since the two measures are regular, for every $\epsilon>0$ there exists $A'\subset A$ closed such that
         \begin{equation}
             |\nu_1|(A\setminus A')<\frac{\epsilon}{2},\qquad  |\nu_2|(A\setminus A')<\frac{\epsilon}{2},
         \end{equation}
         whence
         \begin{align}
             |\nu_1(A)-\nu_2(A)|&=\left|\nu_1(A')+\nu_1(A\setminus A')-\nu_2(A')-\nu_2(A\setminus A')\right|\nonumber\\
             &=\left|\nu_1(A\setminus A')-\nu_2(A\setminus A')\right|\nonumber\\
             &=|\nu_1|(A\setminus A')+|\nu_2|(A\setminus A')<\epsilon.
         \end{align}
         Since $\epsilon$ is arbitrary, we can finally conclude $\nu_1(A)=\nu_2(A)$.
    \end{proof}

    Below, after two preliminary remarks, we will finally present the proof of Theorem~\ref{thm:extension_theorem_pro}, which---like the one presented in the main text---essentially follows step-by-step the second version of the proof of~\cite[Theorem 9.1]{bhattacharya2017basic}. 
    
    \begin{remark}\label{rem:square:brackets}
    When considering complex-valued functions on the product space, $\mathcal{X}:\prod_tM_t\rightarrow\mathbb{C}$, we shall use the following notation:
    \begin{equation}
        \mathcal{X}[f]:=\mathcal{X}\left((f_t)_{t\in I}\right),
    \end{equation}
    i.e.~with the square brackets reminding that we are considering a function, $\mathcal{X}$, of functions (the elements $f\in\prod_tM_t$).
    \end{remark}

     \begin{remark}\label{rem:pushforward}  
    A more compact way to state Eq.~\eqref{eq:cylinders_pro}, which we will use in the proof, is the following.  Given $n\in\mathbb{N},\bm{t}_n\in I^n$, let $\pi_{\bm{t}_n}:\prod_tM_t\rightarrow M_{\bm{t}_n}$ be defined (with the notation of Remark~\ref{rem:square:brackets}) by
    \begin{equation}
        \pi_{\bm{t}_n}[f]:=(f_{t_1},\ldots,f_{t_n}).
    \end{equation}
    Then for every $A\in\mathcal{F}_{\bm{t}_n}$ one has
    \begin{align}
        \pi_{\bm{t}_n}^{-1}(A)&=\left\{f\in\prod_tM_t:\pi_{\bm{t}_n}[f]\in A\right\}\nonumber\\
        &=\left\{f\in\prod_tM_t:(f_{t_1},\ldots,f_{t_n})\in A\right\},
    \end{align}
    and thus the validity of Eq.~\eqref{eq:cylinders_pro} for all Borel sets of $M_{\bm{t}_n}$ corresponds to the equality
    \begin{equation}
        \nu_{\bm{t}_n}=\nu\circ\pi_{\bm{t}_n}^{-1}.
    \end{equation}
    The right-hand side above is the \textit{pushforward} of $\nu$, i.e.~the measure on $M_{\bm{t}_n}$ which, to every Borel set $A$ of $M_{\bm{t}_n}$, associates the quantity $\nu(\pi_{\bm{t}_n}^{-1}(A))$.

    Also notice that $\pi_{\bm{t}_n}$ allows us to ``lift'' any $X_{\bm{t}_n}\in\mathrm{C}(M_{\bm{t}_n})$ to an element of $\mathrm{C}(\prod_tM_t)$ defined by $\mathcal{X}:=X_{\bm{t}_n}\circ\pi_{\bm{t}_n}$, that is, 
    \begin{equation}
        \mathcal{X}[f]:=X_{\bm{t}_n}(\pi_{\bm{t}_n}[f])=X_{\bm{t}_n}(f_{t_1},\ldots,f_{t_n});
    \end{equation}
    in fact, the set of all elements of $\mathrm{C}(\prod_tM_t)$ defined this way will play a fundamental role in the proof below.
    \end{remark}
    
    \begin{proof}[Proof of Theorem~\ref{thm:extension_theorem_pro}]

    To simplify the notation, hereafter we will set
    \begin{equation}
        \mathcal{M}:=\prod_{t\in I}M_t.
    \end{equation}    
    To begin with, since each space $M_{t}$, $t\in I$, is a compact metric space, their product space $\mathcal{M}=\prod_t M_t$, endowed with the product topology, is also compact (and thus a compact Hausdorff space) by Tychonoff's theorem; as such, the space $\mathrm{C}\left(\mathcal{M}\right)$ of continuous functions, endowed with the uniform norm
    \begin{equation}
        \|\mathcal{X}\|_\infty:=\max_{f\in\mathcal{M}}\left|\mathcal{X}[f]\right|,
    \end{equation}    
    is a Banach space. We will construct a bounded linear functional $\mathcal{E}$ on $\mathrm{C}\left(\mathcal{M}\right)$ defined as follows. Let $\mathrm{C}_{\rm fin}\left(\mathcal{M}\right)$ be the space of elements of $\mathrm{C}\left(\mathcal{M}\right)$ that only depend on \textit{finitely many} coordinates, that is: $\mathcal{X}\in\mathrm{C}_{\rm fin}\left(\mathcal{M}\right)$ if and only if there exist $n\in\mathbb{N}$, $\bm{t}_n\in I^n$, and a function $\hat{X}_{\bm{t}_n}\in\mathrm{C}(M_{\bm{t}_n})$ such that
    \begin{equation}\label{eq:finite_repr}
        \mathcal{X}[f]=\tilde{X}_{\bm{t}_n}(f_{t_1},\ldots,f_{t_n}),
    \end{equation}
    or, equivalently, such that $\mathcal{X}=\tilde{X}_{\bm{t}_n}\circ\pi_{\bm{t}_n}$ for some $n\in\mathbb{N}$, $\bm{t}_n\in I^n$, and $\tilde{X}_{\bm{t}_n}\in\mathrm{C}(M_{\bm{t}_n})$.   
    By the Stone--Weierstrass theorem, $\mathrm{C}_{\rm fin}\left(\mathcal{M}\right)$ is a dense subspace of $\mathrm{C}\left(\mathcal{M}\right)$. We now define a linear functional $\mathcal{E}:\mathrm{C}_{\rm fin}\left(\mathcal{M}\right)\rightarrow\mathbb{C}$ as follows: for any $\mathcal{X}\in\mathrm{C}_{\rm fin}\left(\mathcal{M}\right)$ admitting a representation as in Eq.~\eqref{eq:finite_repr}, we define
    \begin{equation}
        \mathcal{E}(\mathcal{X}):=\int\tilde{X}_{\bm{t}_n}(f_{t_1},\ldots,f_{t_n})\;\mathrm{d}\nu_{\bm{t}_n}(f_{t_1},\ldots,f_{t_n}).
    \end{equation}
    We will prove that this functional is (a) well-defined, that is, its value does not depend on the particular representation~\eqref{eq:finite_repr}; and (b) bounded. Point (a) arises from the following observation: a function $\mathcal{X}\in\mathrm{C}_{\rm fin}\left(\mathcal{M}\right)$ admitting the $n$-time representation~\eqref{eq:finite_repr} also admits the $(n+1)$-time representation
    \begin{equation}\label{eq:finite_repr_n-plus-1}
        \mathcal{X}[f]=\tilde{X}_{t_1,\dots,t_n,t_{n+1}}(f_{t_1},\ldots,f_{t_n},f_{t_{n+1}}),
    \end{equation}
    for any $t_{n+1}\in I$, by simply defining $\tilde{X}_{t_1,\dots,t_n,t_{n+1}}(f_{t_1},\ldots,f_{t_n},f_{t_{n+1}}):=\tilde{X}_{t_1,\dots,t_n}(f_{t_1},\ldots,f_{t_n})$. Using this representation, we have
    \begin{align}
        \mathcal{E}(\mathcal{X})&=\int\tilde{X}_{\bm{t}_{n+1}}(f_{t_1},\ldots,f_{t_n},f_{t_{n+1}})\;\mathrm{d}\nu_{\bm{t}_{n+1}}(f_{t_1},\ldots,f_{t_n},f_{t_{n+1}})\nonumber\\
        &=\int\tilde{X}_{\bm{t}_{n}}(f_{t_1},\ldots,f_{t_n})\;\mathrm{d}\nu_{\bm{t}_{n+1}}(f_{t_1},\ldots,f_{t_n},f_{t_{n+1}})\nonumber\\
        &=\int\tilde{X}_{\bm{t}_{n}}(f_{t_1},\ldots,f_{t_n})\;\mathrm{d}\nu_{\bm{t}_{n}}(f_{t_1},\ldots,f_{t_n}),
    \end{align}
    where the latter equality follows by consistency, see Remark~\ref{rem:consistency}. This proves that the functional is well-defined. As for boundedness, we have (cf.~Eq.~\eqref{eq:triangular})
    \begin{align}
        \left|\mathcal{E}(\mathcal{X})\right|&=\left|\int\tilde{X}_{\bm{t}_n}(f_{t_1},\ldots,f_{t_n})\;\mathrm{d}\nu_{\bm{t}_n}(f_{t_1},\ldots,f_{t_n})\right|\nonumber\\
        &\leq\int\left|\tilde{X}_{\bm{t}_n}(f_{t_1},\ldots,f_{t_n})\right|\;\mathrm{d}\left|\nu_{\bm{t}_n}\right|(f_{t_1},\ldots,f_{t_n})\nonumber\\
        &\leq\|\mathcal{X}\|_\infty\left|\nu_{\bm{t}_n}\right|(M_{\bm{t}_n}),
    \end{align}
    whence, maximizing,
    \begin{equation}
        \sup_{\mathcal{X}\in\mathrm{C}_{\rm fin}(\mathcal{M})}\frac{\left|\mathcal{E}(\mathcal{X})\right|}{\|\mathcal{X}\|_\infty}\leq\sup_{n\in\mathbb{N},\bm{t}_n\in I^n}\left|\nu_{\bm{t}_n}\right|(M_{\bm{t}_n})<\infty.
    \end{equation}
        This shows that $\mathcal{E}$ is a \textit{bounded} functional on $\mathrm{C}_{\rm fin}\left(\mathcal{M}\right)$. Since the latter, by the Stone--Weierstrass theorem, is a dense subset of the Banach space $\mathrm{C}\left(\mathcal{M}\right)$, by the continuous linear extension theorem $\mathcal{E}$ admits a unique extension to a bounded functional on $\mathrm{C}\left(\mathcal{M}\right)$, which we will denote by the same symbol. By the Riesz--Markov theorem (Theorem~\ref{thm:riesz-markov}), this implies the existence of a unique regular Borel measure $\nu$ on the product space such that, for every $\mathcal{X}\in\mathrm{C}\left(\mathcal{M}\right)$,
        \begin{equation}\label{eq:functional_integral_pro}
            \mathcal{E}(\mathcal{X})=\int\mathcal{X}[f]\;\mathrm{d}\nu[f].
        \end{equation}
        To conclude the proof, we will need to prove Eq.~\eqref{eq:cylinders_pro}. By construction, the functional $\mathcal{E}$ has the following property: with the notation introduced in Remark~\ref{rem:pushforward}, for any $X_{\bm{t}_n}\in\mathrm{C}(M_{\bm{t}_n})$ we have
    \begin{equation}
        \mathcal{E}(X_{\bm{t}_n}\circ\pi_{\bm{t}_n})=\int X_{\bm{t}_n}(f_{t_1},\ldots,f_{t_n})\;\mathrm{d}\nu_{\bm{t}_n}(f_{t_1},\ldots,f_{t_n}).
    \end{equation}
   But, by Eq.~\eqref{eq:functional_integral_pro}, the left-hand side of the above equation also reads:
   \begin{align}
       \mathcal{E}(X_{\bm{t}_n}\circ\pi_{\bm{t}_n})&=\int(X_{\bm{t}_n}\circ\pi_{\bm{t}_n})[f]\;\mathrm{d}\nu[f]\nonumber\\
       &=\int X_{\bm{t}_n}(f_{t_1},\ldots,f_{t_n})\;\mathrm{d}(\nu\circ\pi_{\bm{t}_n}^{-1})(f_{t_1},\ldots,f_{t_n}).
   \end{align}
   That is, the two Borel measures $\nu_{\bm{t}_n}$ and $\nu\circ\pi_{\bm{t}_n}^{-1}$ on $M_{\bm{t}_n}$ have the following property: their integral over any continuous function $X_{\bm{t}_n}$ is equal. Besides, both measures $\nu_{\bm{t}_n}$ and $\nu\circ\pi_{\bm{t}_n}^{-1}$ are regular (the former because it is defined on a compact metric space, the latter by the regularity of $\nu$ ensured by the Riesz--Markov theorem); by Lemma~\ref{lemma:measure_determining}, this is sufficient to conclude that $\nu_{\bm{t}_n}=\nu\circ\pi_{\bm{t}_n}^{-1}$, and thus (see Remark~\ref{rem:pushforward}) Eq.~\eqref{eq:cylinders_pro} holds, thus concluding the proof.
\end{proof}

\bibliography{bib_extension_theorem}

\end{document}